\newif\ifabstract
\abstractfalse
\newif\iffull
\ifabstract \fullfalse \else \fulltrue \fi
\documentclass[11pt]{article}

\usepackage{amsfonts}
\usepackage{amssymb}
\usepackage{amstext}
\usepackage{amsmath}
\usepackage{xspace}
\usepackage{theorem}
\usepackage{graphicx}
\usepackage{fullpage}
\usepackage{graphics}

\newcommand{\contract}{\mathsf{Contract}}
\newcommand{\betaFCG}{\beta_{\mathrm{FCG}}}

        {\hspace*{\fill}$\Box$\par\vspace{4mm}}

\newcommand{\connect}{\leadsto}
\newcommand{\sconnect}{\overset{\mbox{\tiny{1:1}}}{\leadsto}}
\newcommand{\algsc}{\ensuremath{{\mathcal{A}}_{\mbox{\textup{\scriptsize{ARV}}}}}\xspace}
\newcommand{\alphasc}{\ensuremath{\alpha_{\mbox{\tiny{\sc ARV}}}}}
\newcommand{\ceil}[1]{\ensuremath{\left\lceil#1\right\rceil}}

\newcommand{\set}[1]{\left\{ #1 \right\}}
\newcommand{\sse}{\subseteq}

\newcommand{\tset}{{\mathcal T}}
\newcommand{\eset}{{\mathcal E}}

\newcommand{\pset}{{\mathcal{P}}}
\newcommand{\qset}{{\mathcal{Q}}}

\newcommand{\cset}{{\mathcal{C}}}
\newcommand{\fset}{{\mathcal{F}}}

\newcommand{\wset}{{\mathcal{W}}}

\newcommand{\rset}{{\mathcal{R}}}

\newcommand{\sset}{{\mathcal{{S}}}}
\newcommand{\gset}{{\mathcal{G}}}

\newcommand{\tA}{\tilde A}
\newcommand{\tE}{\tilde E}

\newcommand{\nots}{\overline S}
\newcommand{\be}{\begin{enumerate}}
\newcommand{\ee}{\end{enumerate}}
\newcommand{\bd}{\begin{description}}
\newcommand{\ed}{\end{description}}
\newcommand{\bi}{\begin{itemize}}
\newcommand{\ei}{\end{itemize}}

\newtheorem{lemma}{Lemma}
\newtheorem{theorem}{Theorem}
\newtheorem{observation}{Observation}
\newtheorem{corollary}{Corollary}
\newtheorem{claim}{Claim}

\newtheorem{definition}{Definition}
\iffull \newenvironment{proof}{\par \smallskip{\bf Proof:}}{\hfill\stopproof}
\def\stopproof{\square}
\def\square{\vbox{\hrule height.2pt\hbox{\vrule width.2pt height5pt \kern5pt
\vrule width.2pt} \hrule height.2pt}}
\fi

\renewcommand{\phi}{\varphi}
\newcommand{\eps}{\epsilon}

\newcommand{\half}{\ensuremath{\frac{1}{2}}}

\newcommand{\poly}{\operatorname{poly}}

\newcommand{\reals}{{\mathbb R}}

\newenvironment{properties}[2][0]
{
\begin{enumerate} \setcounter{enumi}{#1}}{\end{enumerate}}

\newcommand{\dmax}{d_{\mbox{\textup{\footnotesize{max}}}}}

\newcommand{\mincut}{\operatorname{MinCut}}
\newcommand{\out}{\operatorname{out}}

\begin{document}

\ifabstract
\conferenceinfo{STOCÕ12,} {May 19--22, 2012, New York, New York, USA.} 
\CopyrightYear{2012} 
\crdata{978-1-4503-1245-5/12/05} 
\clubpenalty=10000 
\widowpenalty = 10000
\fi

\title{On Vertex Sparsifiers with Steiner Nodes}

\ifabstract
\numberofauthors{1}
\author{
\alignauthor Julia Chuzhoy
\titlenote{Supported in part by NSF CAREER grant CCF-0844872 and Sloan Research Fellowship.}\\
       \affaddr{Toyota Technological Institute}\\
       \affaddr{Chicago, IL 60637}\\
       \email{cjulia@ttic.edu}
}
\fi

\iffull
\author{Julia Chuzhoy\thanks{Toyota Technological Institute, Chicago, IL
60637. Email: {\tt cjulia@ttic.edu}. Supported in part by NSF CAREER award CCF-0844872 and by Sloan Research Fellowship}}
\fi

\maketitle
\begin{abstract}
Given an undirected graph $G=(V,E)$ with edge capacities $c_e\geq 1$ for $e\in E$ and a subset $\tset$ of $k$ vertices called terminals, we say that a graph $H$ is a quality-$q$ cut sparsifier for $G$ iff $\tset\sse V(H)$, and for any partition $(A,B)$ of $\tset$, the values of the minimum cuts separating $A$ and $B$ in graphs $G$ and $H$ are within a factor $q$ from each other. We say that $H$ is a quality-$q$ flow sparsifier for $G$ iff $\tset\sse V(H)$, and for any set $D$ of demands over the terminals, the values of the minimum edge congestion incurred by fractionally routing the demands in $D$ in graphs $G$ and $H$ are within a factor $q$ from each other.

So far vertex sparsifiers have been studied in a restricted setting where the sparsifier $H$ is not allowed to contain any non-terminal vertices, that is $V(H)=\tset$. For this setting, efficient algorithms are known for constructing quality-$O(\log k/\log\log k)$ cut and flow vertex sparsifiers, as well as a lower bound of $\tilde{\Omega}(\sqrt{\log k})$ on the quality of any flow or cut sparsifier.

We study flow and cut sparsifiers in the more general setting where Steiner vertices are allowed, that is, we no longer require that $V(H)=\tset$. We show algorithms to construct constant-quality cut sparsifiers of size $O(C^3)$ in time $\poly(n)\cdot 2^C$, and constant-quality flow sparsifiers of size $C^{O(\log\log C)}$ in time $n^{O(\log C)}\cdot 2^C$, where $C$ is the total capacity of the edges incident on the terminals.
\end{abstract}

%-------------------------------------------------------------------------------------------
%-------------------------------------------------------------------------------------------
%-------------------------------------------------------------------------------------------
\section{Introduction}
%-------------------------------------------------------------------------------------------
%-------------------------------------------------------------------------------------------
%-------------------------------------------------------------------------------------------
Suppose we are given an undirected graph $G=(V,E)$ with edge capacities $c_e\geq 1$ for $e\in E$, and a subset $\tset\sse V$ of $k$ vertices called terminals. Assume further that we are interested in routing traffic across  $G$ between the terminals in $\tset$. While the size of the graph $G$ may be very large, the specific structure of $G$ is largely irrelevant to our task, except where it affects our ability to route flow between the terminals. A natural question is whether we can build a smaller graph $H=(V',E')$, with $\tset\sse V'$, that  approximately preserves the routing properties of graph $G$ with respect to $\tset$. In this case, we say that $H$ is a \emph{vertex sparsifier} for $G$. Two types of vertex sparsifiers have been studied so far: cut sparsifiers, which preserve the minimum cuts between any partition of the terminals, and flow sparsifiers, which preserve the  minimum edge congestion required for routing any set $D$ of demands over $\tset$. 

More formally,
given any graph $G$ with capacities $c_e\geq 1$ for the edges $e\in E$, a subset $\tset$ of vertices called terminals, and a partition $(\tset_A,\tset_B)$ of the terminals, let $\mincut_{G}(\tset_A,\tset_B)$ denote the capacity of the minimum cut separating the vertices of $\tset_A$ from the vertices of $\tset_B$ in $G$. We say that a graph $H=(V',E')$ is a quality-$q$ \emph{vertex cut sparsifier}, or just cut sparsifier, for graph $G$ with terminal set $\tset$, iff $\tset\sse V'$, and for every partition $(\tset_A,\tset_B)$ of $\tset$, $\mincut_G(\tset_A,\tset_B)\leq \mincut_H(\tset_A,\tset_B)\leq q\cdot \mincut_G(\tset_A,\tset_B)$.

Given a graph $G$ with capacities $c_e\geq 1$ for every edge $e\in E$, and a subset $\tset\sse V$ of vertices called terminals, a set $D$ of demands over the terminals specifies, for every unordered pair $(t,t')$ of terminals, a demand $D(t,t')$. A flow $F$ is a \emph{routing} of the set $D$ of demands, iff for every pair $(t,t')$ of terminals, $t$ and $t'$ send $D(t,t')$ flow units to each other. The congestion of $F$ is the maximum, over all edges $e\in E$, of $F(e)/c_e$, where $F(e)$ is the flow sent along $e$. Given a set $D$ of demands over the set $\tset$ of terminals, let $\eta(G,D)$ denote the minimum congestion required for routing the demands in $D$ in graph $G$.
We say that a graph $H$ is a \emph{flow sparsifier of quality $q$} for $G$, iff $\tset\sse V(H)$, and for any set $D$ of demands over the set $\tset$ of terminals, $\eta(H,D)\leq \eta(G,D)\leq q\cdot \eta(H,D)$. 

For a vertex sparsifier $H$, we say that the vertices in $V(H)\setminus \tset$ are \emph{Steiner vertices}. 
Vertex cut sparsifiers were first introduced by Moitra~\cite{Moitra}, and later Leighton and Moitra~\cite{LM} defined flow sparsifiers and showed that they generalize cut sparsifiers. The main motivation in both papers was designing improved approximation algorithms for graph partitioning and routing problems. Specifically, if the solution value of some combinatorial optimization problem only depends on the values of the minimum cuts separating terminal subsets, then given any approximation algorithm for the problem, we can first compute a cut sparsifier $H$ for graph $G$, and then run this algorithm on $H$, thus obtaining an algorithm whose performance guarantee is independent of the size of $G$, and only depends on the size of the sparsifier $H$. Flow sparsifiers can be similarly used for combinatorial optimization problems whose solution value only depends on the congestion required for routing various demand sets over the terminals in $G$. The definitions of the cut and the flow sparsifiers of \cite{Moitra,LM} however required that the sparsifier $H$ does not contain any Steiner vertices, that is, $V(H)=\tset$.

Moitra~\cite{Moitra} showed that there exist  cut sparsifiers of quality $O(\log k/\log\log k)$ even when no Steiner vertices are allowed, and Leighton and Moitra~\cite{LM} proved the existence of quality-$O(\log k/\log\log k)$ flow sparsifiers for the same setting, and obtained an efficient algorithm to construct quality-$O(\log^2k/\log\log k)$ flow and cut sparsifiers. Recently, Charikar et al.~\cite{CLLM}, Englert et al.~\cite{EGK} and Makarychev and Makarychev~\cite{MM} have shown efficient algorithms to construct quality-$O(\log k/\log\log k)$ flow and cut sparsifiers that do not contain Steiner vertices. On the negative side, Leighton and Moitra~\cite{LM} have shown a
 lower bound of $\Omega(\log\log k)$ on the quality of flow sparsifiers when no Steiner vertices are allowed. This  bound was later improved to $\Omega(\sqrt{\log k/\log\log k})$ by Makarychev and Makarychev~\cite{MM}. Englert et al.~\cite{EGK} have shown a lower bound of $\Omega(\sqrt {\log k}/\log\log k)$ on the quality of flow sparsifiers when no Steiner vertices are allowed, and all edge capacities of the sparsifier are bounded from below by a constant. As for cut vertex sparsifiers with no Steiner nodes, \cite{CLLM} and \cite{MM} have shown a lower bound of $\Omega(\log^{1/4}k)$, and the results of~\cite{MM} together with the results of~\cite{FJS} give a lower bound of $\Omega(\sqrt{\log k}/\log\log k)$ on their quality.

It is therefore natural to ask whether we can obtain better quality vertex sparsifiers by allowing Steiner vertices. In particular, an interesting question is: what is the smallest size $S(k)$, such that for any graph $G$ with a set $\tset$ of $k$ terminals, there is a constant-quality cut or flow sparsifier of size at most $S(k)$. Notice that if our goal is to obtain better and faster approximation algorithms via graph sparsifiers, then the presence of Steiner nodes may actually lead to improved performance if we can construct better quality sparsifiers, while keeping the graph size sufficiently low.

For simplicity, we first consider a special case where all edge capacities in the input graph $G$ are unit, and every terminal in $\tset$ has degree $1$ in $G$. In this case we show that there exist constant-quality cut sparsifiers of size $O(k^3)$, and constant-quality flow sparsifiers of size $k^{O(\log\log k)}$. We also show algorithms to construct these sparsifiers in time $\poly(n)\cdot 2^k$ for cut sparsifiers and in time $n^{O(\log k)}\cdot 2^k$ for flow sparsifiers. We then generalize these algorithms to arbitrary edge capacities. Let $C$ be the total capacity of the edges incident on the terminals, assuming that for each edge $e\in E$, $c_e\geq 1$. We show that there exist constant-quality cut sparsifiers of size $O(C^3)$, and constant-quality flow sparsifiers of size $C^{O(\log\log C)}$, and show algorithms to construct such sparsifiers, with running time $\poly(n)\cdot 2^C$ for cut sparsifiers and $n^{O(\log C)}\cdot 2^C$ for flow sparsifiers.

We say that a graph $H$ is a \emph{restricted sparsifier} for graph $G$, if $H$ is a sparsifier that is associated with a collection $\cset$ of disjoint subsets of non-terminal vertices, and $H$ is obtained from $G$ by contracting every cluster $S\in \cset$ into a vertex. All sparsifiers that we construct are restricted sparsifiers. Interestingly, Charikar et al.~\cite{CLLM} showed that when Steiner vertices are not allowed, the ratio of the quality of the best possible restricted flow sparsifier to the quality of an optimal flow sparsifier is super-constant. Moreover, Englert et al.~\cite{EGK} have shown an $\Omega(\sqrt {\log k})$ lower bound on the quality of sparsifiers that do not contain Steiner vertices, and can be obtained from convex combinations of $0$-extensions in graph $G$.

We note that our techniques are very different from the techniques of~\cite{MM,CLLM,EGK}, who exploited the connection between vertex sparsifiers and $0$-extensions. Instead, we use well-linked decompositions and other techniques that are often employed in the context of graph routing.

\paragraph{Our Results and Techniques}

\iffalse

\begin{definition}
Let $G=(V,E)$ be any graph, and let $\tset\sse V$ be a subset of vertices, called terminals. We say that a graph $H=(V',E')$ is a quality-$q$ vertex cut sparsifier for $G$, iff $\tset\sse V'$, and for every partition $(\tset_A,\tset_B)$ of $\tset$, $\mincut_G(\tset_A,\tset_B)\leq \mincut_H(\tset_A,\tset_B)\leq q\cdot \mincut_G(\tset_A,\tset_B)$.
\end{definition}
\fi
We start with a  simple construction of cut sparsifiers with Steiner vertices, which is summarized in the following theorem.

\begin{theorem}\label{thm: cut sparsifiers}
Let $G=(V,E)$ be any $n$-vertex graph with capacities $c_e\geq 1$ on edges $e\in E$, and a set $\tset\sse V$ of terminals. Let $C$ denote the total capacity of the edges incident on the terminals, and let $0<\eps\leq 1$ be any constant. Then there is a quality-$(3+\eps)$ vertex cut sparsifier $H=(V',E')$ for $G$, with $|V'|=O(C^3)$. Moreover, graph $H$ can be constructed in time $\poly(n)\cdot 2^C$.
\end{theorem}

\iffalse
Given a graph $G=(V,E)$ with edge capacities $c_e>0$ for each $e\in E$, a subset $\tset\sse V$ of vertices called terminals, and a set $D$ of demands on the terminals, let $\eta(G,D)$ denote the minimum congestion required for (fractionally) routing the demands in $D$ in graph $G$.
\fi

For simplicity, we give an outline of the construction for the special case where all edge capacities are unit, and the degree of every terminal is $1$.
Our algorithm relies on the notion of well-linkedness, and on a new procedure to compute a well-linked decomposition. Given any subset $S$ of vertices, let $\out(S)$ denote the set of edges with exactly one endpoint  in $S$. We say that $S$ is $\alpha$-well-linked, iff for any partition $(A,B)$ of $S$, if we denote $\tset_A=\out(S)\cap \out(A)$ and $\tset_B=\out(S)\cap \out(B)$, then $|E(A,B)|\geq \alpha\cdot \min\set{|\tset_A|,|\tset_B|}$. Informally, we can set up an instance of the sparsest cut problem, on graph $G[S]\cup \out(S)$, where the edges of $\out(S)$ serve as terminals. Set $S$ being $\alpha$-well-linked is roughly equivalent to the value of sparsest cut in this new graph being at least $\alpha$. The notion of well-linkedness\footnote{Our definition of well-linkedness is very similar to what was called bandwidth property in~\cite{Raecke}, and cut well-linkedness in~\cite{CKS}, where we use the graph $G[S]\cup \out(S)$, and the set of terminals is the edges of $\out(S)$.} has been used extensively in graph routing~ e.g. in~\cite{Raecke,ANF,CKS,RaoZhou,Andrews}, and one of the useful tools for designing algorithms for routing problems is  \emph{well-linked decomposition}: a procedure that, given any subset $S$ of vertices with $|\out(S)|=z$, produces a partition $\wset$ of $S$ into well-linked subsets. In all standard well-linked decompositions, we can ensure that $|\wset|$ is  small (less than $z$), while each set $X\in \wset$ is guaranteed to be $\alpha$-well-linked, where $\alpha=1/\poly\log z$. We show a different well-linked decomposition, that instead ensures that every set $X\in \wset$ is $1/3$-well-linked, and we can still bound the number of clusters in $\wset$ by $O(z^3)$. An algorithm for constructing a cut sparsifier then simply computes a well-linked decomposition $\wset$ of the set $V(G)\setminus\tset$ of vertices, and contracts every cluster $X\in \wset$. Since every cluster $X\in \wset$ is $1/3$-well-linked, it is easy to verify that we obtain a constant-quality cut sparsifier.

We now turn to the more challenging task of constructing flow sparsifiers. We again first consider a special case where all edge capacities are unit, and each terminal $t\in \tset$ has exactly one edge incident to it in $G$. We show that for this special case, there is a flow sparsifier $H$ of quality $68$ and size $k^{O(\log\log k)}$, where $k=|\tset|$. Recall that a sparsifier $H$ is called a restricted sparsifier iff it is associated with a collection $\cset$ of disjoint subsets of non-terminal vertices, and graph $H$ is obtained from $G$ by contracting each cluster $S\in\cset$ into a vertex.

\iffalse
\begin{definition}
Let $G=(V,E)$ be any graph with a set $\tset\sse V$ of terminals. We say that a graph $H$ is a \emph{flow sparsifier of quality $q$} for $G$, iff $\tset\sse V(H)$, and for any set $D$ of demands on the set $\tset$ of terminals, $\eta(H,D)\leq \eta(G,D)\leq q\cdot \eta(H,D)$.

We say that $H$ is a \emph{restricted} sparsifier for $G$ iff there is a collection $\cset$ of disjoint subsets of {\bf non-terminal} vertices of $G$, such that $H$ is obtained from $G$ by contracting each cluster $C\in \cset$ into a super-node $v_C$. (We remove the self-loops but we do not remove parallel edges in the resulting graph).
\end{definition}

\fi

\begin{theorem}\label{thm: flow sparsifier for unit capacities}
Let $G=(V,E)$ be any $n$-vertex (multi-)graph with unit edge capacities and a set $\tset\sse V$ of $k$ terminals. Assume further that each terminal in $\tset$ has exactly one edge incident to it in $G$. Then there is an algorithm that finds, in time $n^{O(\log k)}\cdot 2^k$, a quality-$q$ restricted vertex flow sparsifier $H$ for $G$, with $|V(H)|=k^{O(\log\log k)}$ and $q=68$.
\end{theorem}

It is then fairly easy to obtain the following corollary that extends the results of Theorem~\ref{thm: flow sparsifier for unit capacities} to general graphs. 

\begin{corollary}\label{corollary: general flow sparsifier} \label{COROLLARY: GENERAL FLOW SPARSIFIER}
Let $G=(V,E)$ be any $n$-vertex graph with edge capacities $c_e\geq 1$ for $e\in E$, and a set $\tset\sse V$ of  terminals. Let $C$ denote the total capacity of all edges incident on the terminals, and let $0<\eps<1$ be any constant. Then there is an algorithm that finds, in time $n^{O(\log C)}\cdot 2^C$, a quality-$q$ vertex flow sparsifier $H$ for $G$, with $|V(H)|=C^{O(\log\log C)}$ and $q=68+\eps$. %If all edge capacities in $G$ are integers of value at most $C$, then we can ensure that $H$ is a restricted sparsifier.
\end{corollary}

 We now outline our algorithm for constructing flow sparsifiers for the special case where all edge capacities are unit, and the degree of every terminal is $1$. Let us assume for simplicity that the set $R=V(G)\setminus \tset$ of vertices is $1/3$-well-linked (we perform a well-linked decomposition as a pre-processing step to ensure this). One of the central notions in our algorithm is that of good routers. We say that a subset $S\sse R$ of vertices is a good router iff it is $1/3$-well-linked, and moreover, every pair of edges in $\out(S)$ can simultaneously send $1/z$ flow units to each other with constant congestion inside $S$, where $z=|\out(S)|$. We say that a graph $H$ is a legal contracted graph for $G$ iff there is a collection $\cset$ of disjoint good routers in graph $G$, and $H$ is obtained from $G$ by contracting every cluster $S\in \cset$. It is easy to verify that if $H$ is a legal contracted graph, then it is a constant quality flow sparsifier, since contracting the good routers in $\cset$ may only affect the congestion of any routing by a constant factor. Our goal is then to find a legal contracted graph whose size is small enough.
 
 Notice that we have assumed that $R=V(G)\setminus\tset$ is $1/3$-well-linked. However, this is not sufficient to ensure that $R$ is a good router, as the ratio between the minimum sparsest cut and the maximum concurrent flow, known as the flow-cut gap, can be as large as logarithmic in undirected graphs. To overcome this difficulty, we define several special structures that we call witnesses. If graph $G$ contains such a witness, then we are guaranteed that $R$ is a good router. For example, suppose that for some value $0<\alpha<1$, graph $G$ contains $r=\log k/\alpha$ disjoint subsets $S_1,\ldots,S_r$ of non-terminal vertices, where for each $1\leq j\leq r$, subset $S_j$ is $\alpha$-well-linked, and there is a set $\pset_j$ of edge-disjoint paths in $G$, connecting every  terminal in $\tset$ to some edge in $\out(S_j)$. For each $1\leq j\leq r$, let $E_j\sse \out(S_j)$ be the set of $k$ edges where the paths of $\pset_j$ terminate.
Since the flow-cut gap in undirected graphs is bounded by $O(\log k)$, and the set $S_j$ is $\alpha$-well-linked,  every pair of edges in $E_j$ can simultaneously send $\frac{\alpha}{k\log k}=\frac{1}{rk}$ flow units to each other with constant congestion inside $S_j$. If graph $G$ contains such a witness $\set{S_1,\ldots,S_r}$, it is easy to verify that $R$ must be a good router, since we can send, for each $1\leq j\leq r$, $1/r$ flow units along each path in $\pset_j$, so that each edge $e\in \out(S_j)$ receives at most $1/r$ flow units, and then send $\frac{1}{rk}$ flow units between every pair of edges in  $E_j$, with constant congestion inside $S_j$. In this way, every pair of terminals sends $1/k$ flow units to each other with constant congestion in $G$.
 
 Our algorithm proceeds as follows. Throughout the algorithm, we maintain a legal contracted graph $G'$ of $G$, where at the beginning $G'=G$. As long as the number of vertices in $G'$ is large enough, we perform an iteration, whose output is either a witness for set $R$ being a good router, or another legal contracted graph $G''$ that contains fewer vertices than $G'$. In the former case, we stop the algorithm and output a sparsifier $H$ obtained from $G$ by contracting the set $R$, and in the latter case we proceed to the next iteration. In fact, we can efficiently check whether $R$ is a good router beforehand, by computing an appropriate multicommodity flow in $G[R]$, to ensure that the former case never happens. Once the size of the current graph $G'$ becomes small enough, we output it as our final sparsifier.
 
{\bf Organization:} We start with preliminaries and notation in Section~\ref{sec: Prelims}. We construct cut sparsifiers in Section~\ref{sec:cut sparsifiers} and flow sparsifiers in Section~\ref{sec: flow sparsifiers}.% For convenience, a list of parameters appears in Section~\ref{sec: parameters list} of the Appendix.

%-------------------------------------------------------------------------------------------
%-------------------------------------------------------------------------------------------
%-------------------------------------------------------------------------------------------
\section{Preliminaries and Notation}
%-------------------------------------------------------------------------------------------
%-------------------------------------------------------------------------------------------
%-------------------------------------------------------------------------------------------
\label{--------------------------------sec: prelims--------------------------------------}
\label{sec: Prelims}
\label{SEC: PRELIMS}
In all our results, we start with a special case where all edges in $G$ have unit capacities, and then extend our results to the general setting. Therefore, all definitions and results presented in this section are for graphs with unit edge capacities.

\paragraph{General Notation}
For a graph $G=(V,E)$, and subsets $V'\sse V$, $E'\sse E$ of its vertices and edges respectively, we denote by $G[V']$, $G\setminus V'$, and $G\setminus E'$ the sub-graphs of $G$ induced by $V'$, $V\setminus V'$, and $E\setminus E'$, respectively.
For any subset $S\sse V$ of vertices, we denote by $\out_G(S)=E_G(S,V\setminus S)$ the subset of edges with one endpoint in $S$ and the other endpoint in $V\setminus S$. When clear from context, we omit the subscript $G$. All logarithms are to the base of $2$.
%Throughout the paper, we say that a random event succeeds w.h.p., if the probability of success is $(1-1/\poly(n))$, where $n$ is the size of the input graph.

%Suppose we are given a collection $\cset$ of disjoint vertex subsets of $G$, where each $C\in \cset$ only contains non-terminal vertices. We denote by $H=\contract(G,\cset)$ the graph obtain from $G$, after we contract each cluster $C\in \cset$ into a super-node $v_C$. We say that $H$ is a \emph{contraction} of $G$. %All cut and flow sparsifiers that we obtain are contractions of the original graph $G$.

Let $\pset$ be any collection of paths in graph $G$. We say that paths in $\pset$ cause congestion $\eta$ in $G$, iff for each edge $e\in E(G)$, the number of paths in $\pset$ containing $e$ is at most $\eta$.

Given a graph $G=(V,E)$, and a subset $\tset\sse V$ of vertices called terminals, a set $D$ of demands is a function $D:\tset\times \tset\rightarrow \reals^+$, that specifies, for each pair $t,t'\in \tset$ of terminals, a demand $D(t,t')$. For simplicity, we assume that the pairs $t,t'$ of terminals are unordered, that is $D(t,t')=D(t',t)$ for all $t,t'\in\tset$.
We say that the set $D$ of demands is \emph{$\gamma$-restricted}, iff for each terminal $t\in \tset$, the total demand $\sum_{t'\in \tset}D(t,t')\leq \gamma$.
%We say that the set $D$ of demands is \emph{integral} iff $D_{t,t'}$ is an integer for each $t,t'\in \tset$.

Given any set $D$ of demands, a \emph{routing} of $D$ is a flow $F$, where for each unordered pair $t,t'\in \tset$, the amount of flow sent from $t$ to $t'$ (or from $t'$ to $t$) is $D(t,t')$. The \emph{congestion} of the flow is the maximum, over all edges $e\in E$, of $F(e)$ --- the amount of flow sent via the edge $e$.
% Given an integral set $D$ of demands, an \emph{integral} routing of $D$ is a collection $\pset$ of paths, where for each unordered pair $(t,t')\in \tset$, there are $D_{t,t'}$ paths connecting $t$ to $t'$ in $\pset$. The congestion of this integral routing is the congestion caused by the set $\pset$ of paths in $G$.
%Any matching $M$ on the set $\tset$ of terminals defines an integral $1$-restricted set $D$ of demands, where $D(t,t')=1$ if $(t,t')\in M$, and $D(t,t')=0$ otherwise. We will not distinguish between the matching $M$ and the corresponding set $D$ of demands.

Given any two subsets $V_1,V_2$ of vertices, we denote by $F: V_1\connect_{\eta}V_2$ a flow that causes congestion at most $\eta$ in $G$, where each vertex in $V_1$ sends one flow unit, and each flow-path starts at a vertex of $V_1$ and terminates at a vertex of $V_2$. We denote by $F: V_1\sconnect_{\eta}V_2$ a flow with the above properties, where additionally each vertex in $V_2$ receives at most one flow unit. Similarly, we denote by $\pset:V_1\connect_{\eta}V_2$ a collection of paths $\pset=\set{P_v\mid v\in V_1}$ in graph $G$, where each path $P_v$ originates at $v$ and terminates at some vertex of $V_2$, and the paths in $\pset$ cause congestion at most $\eta$. We denote $\pset:V_1\sconnect_{\eta}V_2$ if additionally each vertex of $V_2$ serves as an endpoint of at most one path in $\pset$. Similarly, we define flows and paths between subsets of edges. For example, given two collections $E_1,E_2$ of edges of $G$, we denote by $F:E_1\connect_{\eta}E_2$ a flow that causes congestion at most $\eta$ in $G$, where each flow-path has an edge in $E_1$ as its first edge, and an edge in $E_2$ as its last edge, and moreover each edge in $E_1$ sends one flow unit. (Notice that it is then guaranteed that each edge in $E_2$ receives at most $\eta$ flow units due to the bound on congestion). If additionally each edge in $E_2$ receives at most one flow unit, we denote this by $F:E_1\sconnect_{\eta}E_2$. Collections of paths connecting subsets of edges to each other are defined similarly. We will often be interested in a scenario where we are given a subset $S\sse V(G)$ of vertices, and $E_1,E_2\sse \out(S)$. In this case, we say that a flow $F:E_1\connect_{\eta}E_2$ is \emph{contained} in $S$, iff for each flow-path $P$ in $F$, all edges of $P$ belong to $G[S]$, except for the first and the last edges that belong to $\out(S)$. Similarly, we say that a set $\pset:E_1\connect_{\eta}E_2$ of paths is contained in $S$, iff all inner edges on paths in $\pset$ belong to $G[S]$.

\paragraph{Sparsest Cut and the Flow-Cut Gap}
 
  Suppose we are given a graph $G=(V,E)$, with non-negative weights $w_v$ on vertices $v\in V$, and a subset $\tset\sse V$ of $k$ terminals, such that for all $v\not\in \tset$, $w_v=0$. Given any partition $(A,B)$ of $V$, the \emph{sparsity} of the cut $(A,B)$ is $\frac{|E(A,B)|}{\min\set{W(A),W(B)}}$, where $W(A)=\sum_{v\in A}w_v$ and $W(B)=\sum_{v\in B}w_v$. In the sparsest cut problem, the input is a graph $G$ with non-negative weights on vertices, and the goal is to find a cut of minimum sparsity. Arora, Rao and Vazirani~\cite{ARV} have shown an $O(\sqrt{\log k})$-approximation algorithm for the sparsest cut problem. We denote by $\algsc$ this algorithm and by $\alphasc(k)=O(\sqrt{\log k})$ its approximation factor.
We will usually work with a special case of the sparsest cut problem, where for each $t\in \tset$, $w_t=1$. We denote such an instance by $(G,\tset)$.

The dual of the sparsest cut problem is the maximum concurrent flow problem, where the goal is to find the maximum possible value $\lambda$, such that each pair $(t,t')$ of terminals can simultaneously send $\lambda/k$ flow units to each other with unit congestion (we assume that in the sparsest cut problem instance the weights $w_t=1$ for all $t\in \tset$). 
The flow-cut gap is the maximum possible ratio, in any graph, between the value of the minimum sparsest cut and the value $\lambda$ of the maximum concurrent flow. The flow-cut gap in undirected graphs, that we denote by $\betaFCG(k)$ throughout the paper, is $\Theta(\log k)$~\cite{LR, GVY,LLR,Aumann-Rabani}. In particular, if the value of the sparsest cut in graph $G$ is $\alpha$, then every pair of terminals can send at least $\frac{\alpha}{k\betaFCG(k)}$ flow units to each other simultaneously with no congestion. It is also easy to see that any $1$-restricted set $D$ of demands on set $\tset$ of terminals can be routed with congestion at most $2\betaFCG(k)/\alpha$. In order to find this routing, let $F$ be the flow where every pair of terminals sends $\frac{\alpha}{k\betaFCG(k)}$ flow units to each other with no congestion, and let $F'$ be the same flow scaled up by factor $\betaFCG(k)/\alpha$, so the flow in $F'$ causes congestion at most $\betaFCG(k)/\alpha$, and every pair of terminals sends $1/k$ flow units to each other. For each pair $(t,t')$ of terminals, vertex $t$ sends $D(t,t')/k$ flow units to each terminal in $\tset$ using the flow $F'$ (scaled by factor $D(t,t')$), and vertex $t'$ collects $D(t,t')/k$ flow units from each terminal in $\tset$. It is easy to verify that, since the set $D$ of demands is $1$-restricted, the total congestion of this flow is bounded by $2\betaFCG(k)/\alpha$.

\paragraph{Well-Linked Decompositions}
%--------------------------------------------------
\begin{definition} Given a graph $G$, a subset $S$ of its vertices, and a parameter $\alpha>0$, we say that $S$ is $\alpha$-well-linked, iff for any partition $(A,B)$ of $S$, if we denote by $T_A=\out(A)\cap \out(S)$, and by $T_B=\out(B)\cap \out(S)$, then
$|E(A,B)|\geq \alpha\cdot \min\set{|T_A|,|T_B|}$.
\end{definition}

Given a subset $S$ of vertices of $G$, we define a graph $G_S$ associated with $S$, and a corresponding instance $(G_S,\tset'_S)$ of the sparsest cut problem, that we use throughout the paper. We start by sub-dividing every edge $e\in \out_G(S)$ by a vertex $t_e$, and let $\tset'_S=\set{t_e\mid e\in \out_G(S)}$ be the set of these new vertices. We then let $G_S$ be the sub-graph of the resulting graph, induced by $S\cup \tset'_S$. Notice that set $S$ is $\alpha$-well-linked in $G$ iff the value of the sparsest cut in instance $(G_S,\tset'_S)$ is at least $\alpha$ (for $\alpha<1$). In particular, if $S$ is $\alpha$-well-linked, and $|\out(S)|=z$, then we have the following two properties:

%\vspace{-4mm}
\label{---------------------------------Properties of well-linked------------------------------}
\begin{properties}{P}
\item Any set $D$ of $1$-restricted demands on the edges of $\out(S)$ can be routed inside $S$ with congestion at most $2\betaFCG(z)/\alpha$. \label{Property: any matching can be routed with low congestion}
\item For any two subsets $E_1,E_2\sse \out(S)$, where $|E_1|=|E_2|$, there is a collection $\pset:E_1\sconnect_{\lceil 1/\alpha\rceil}E_2$ of paths contained in $S$. \label{Property: any partition of edges can be routed}
\end{properties}
\label{---------------------------------End Properties of well-linked------------------------------}
%\vspace{-4mm}

In order to obtain the latter property, we  set up a single-source single-sink max-flow instance in graph $G_S$, where the edges of $E_1$ serve as the source and the edges of $E_2$ serve as the sink. The existence of the flow $F:E_1\sconnect_{1/\alpha}E_2$ follows from the max-flow/min-cut theorem, and the existence of the set $\pset:E_1\sconnect_{\lceil 1/\alpha\rceil}E_2$ of paths follows from the integrality of flow.

A well-linked decomposition of an arbitrary subset $S$ of vertices, is a partition of $S$ into a collection of  well-linked subsets. We use two different types of well-linked decomposition, that give slightly different guarantees. We start with a standard decomposition, that we refer to as the \emph{weak well-linked decomposition}, and it is similar to the one used in \cite{CKS, Raecke}. The proof of the next theorem appears in the Appendix.

\begin{theorem}[Weak well-linked decomposition]\label{thm: weak well-linked} Given any graph $G=(V,E)$, and any subset $S\sse V$ of vertices with $|\out(S)|=z$, there is an efficient algorithm, that finds a partition $\wset$ of $S$, such that for each set $R\in \wset$, $|\out(R)|\leq |\out(S)|$, $R$ is $\alpha_W(z)$-well-linked for $\alpha_W(z)=\Omega\left(\frac 1 {\log^{3/2} z}\right)$, and $\sum_{R\in \wset}|\out(R)|\leq 1.2|\out(S)|$.
\end{theorem}

The next theorem gives what we call a \emph{strong well-linked decomposition}. This decomposition gives a better guarantee for the well-linkedness of the resulting sets in the partition. The drawback is that the running time of the algorithm is exponential in $|\out(S)|$, and the number of edges adjacent to the subsets in the partition is higher.
The proof of the next theorem appears in the Appendix.

\begin{theorem} [Strong well-linked decomposition] \label{thm: strong well linked}
Given any $n$-vertex graph $G=(V,E)$, and any subset $S\sse V$ of its vertices, where $G[S]$ is connected and $|\out(S)|=z$, there is an algorithm running in time $2^z\cdot \poly(n)$, that finds a partition $\sset$ of $S$, such that:

\begin{itemize}
\item For each $R\in \sset$, $|\out(R)|\leq |\out(S)|$, and $R$ is $1/3$-well-linked;

\item $\sum_{R\in \sset}|\out(R)|=O(z^3)$; and

\item For each $i: 1\leq i\leq 2\lfloor \log z\rfloor$, if $\sset_i\subseteq \sset$ denotes the collection of subsets $R\in \sset$ with $z/2^{i}< |\out(R)|\leq z/2^{i-1}$, then $|\sset_i|\leq 2^{3i+3}$ for all $i$.
\end{itemize}
\end{theorem}

We will sometimes use the notion of well-linkedness in a slightly different setting. Suppose we are given a graph $G=(V,E)$, and a subset $\tset\sse V$ of vertices called terminals. We say that $G$ is $\alpha$-well-linked with respect to $\tset$, iff for any partition $(A,B)$ of $V$, if we denote $T_A=\tset\cap A$ and $T_B=\tset\cap B$, then $|E(A,B)|\geq \alpha\cdot \min\set{|T_A|,|T_B|}$. A convenient way of viewing this consistently with the previous definition of well-linkedness is to augment the graph $G$, by adding an edge connecting each terminal $t\in \tset$ to a new vertex $v_t$. Saying that $G$ is $\alpha$-well-linked for $\tset$ is then equivalent to saying that the subset $V$ of vertices of the new graph is $\alpha$-well-linked.

\label{------------------------------------------Cut Sparsifiers--------------------------------------------------------}
%-------------------------------------------------------------------------------------------
%-------------------------------------------------------------------------------------------
%-------------------------------------------------------------------------------------------
\section{Cut Sparsifiers}\label{sec:cut sparsifiers}
%-------------------------------------------------------------------------------------------
%-------------------------------------------------------------------------------------------
%-------------------------------------------------------------------------------------------

%-------------------------------------------------------------------------------------------
%-------------------------------------------------------------------------------------------
%-------------------------------------------------------------------------------------------
%\subsection{Cut Sparsifiers}\label{subsec: cut sparsifiers}
%-------------------------------------------------------------------------------------------
%-------------------------------------------------------------------------------------------
%-------------------------------------------------------------------------------------------

In this section we prove Theorem~\ref{thm: cut sparsifiers}. We fist consider a simpler special case where all edge capacities are unit (but parallel edges are allowed), in the following theorem.

\begin{theorem}\label{thm: cut sparsifiers - unweighted}
Let $G=(V,E)$ be any $n$-vertex (multi-)graph with unit edge capacities, and a set $\tset\sse V$ of terminals. Let $k=\sum_{t\in \tset}d_t$ be the sum of degrees of all terminals. Then there is a quality-$3$ vertex cut sparsifier $H=(V',E')$ for $G$, with $|V'|=O(k^3)$. Moreover, graph $H$ can be constructed in time $\poly(n)\cdot 2^k$.
\end{theorem}

\begin{proof}
We assume w.l.o.g. that $G$ is a connected graph: otherwise, we construct a sparsifier for each of its connected components separately.
Let $S=V\setminus \tset$. Notice that $|\out(S)|=k$. In order to construct the sparsifier $H$, we compute a strong well-linked decomposition $\sset$ of $S$, given by Theorem~\ref{thm: strong well linked}. Recall that the decomposition can be found in time $2^k\cdot \poly(n)$, and $|\sset|=O(k^3)$. We now contract each set $R\in \sset$ into a single super-node $v_R$. The resulting graph is the sparsifier $H$. Notice that $H$ is an unweighted multi-graph, and $|V(H)|=O(k^3)$.

Assume that we are given any partition $(\tset_A,\tset_B)$ of the set $\tset$ of terminals.
It is easy to see that $\mincut_G(\tset_A,\tset_B)\leq \mincut_H(\tset_A,\tset_B)$: let $(X',Y')$ be the minimum cut, separating $\tset_A$ from $\tset_B$ in graph $H$. Since $H$ is obtained from $G$ by contracting some subsets of  its vertices, the cut $(X',Y')$ naturally induces a cut $(X,Y)$ separating $\tset_A$ from $\tset_B$ in $G$: for each cluster $R\in \sset$, if $v_R\in X'$, then we add all vertices of $R$ to $X$, and otherwise we add them to $Y$. The value of the cut, $|E_G(X,Y)|=|E_H(X',Y')|$, and so $\mincut_G(\tset_A,\tset_B)\leq \mincut_H(\tset_A,\tset_B)$.

We now prove that $\mincut_H(\tset_A,\tset_B)\leq 3\mincut_G(\tset_A,\tset_B)$. Let $(X,Y)$ be the minimum cut separating $\tset_A$ from $\tset_B$ in $G$. We define a cut $(X',Y')$, separating $\tset_A$ from $\tset_B$ in $H$, with $|E_H(X',Y')|\leq 3|E_G(X,Y)|$, as follows. we start with the cut $(X,Y)$ in graph $G$, and we gradually change this cut, so that eventually, for each set $R\in \sset$, all vertices of $R$ are completely contained in either $X$ or in $Y$. The resulting partition will then naturally define the cut $(X',Y')$ in graph $H$.

We process the sets $R\in \sset$ one-by-one. Let $R$ be any such set. Partition the edges of $\out(R)$ into four subsets: $E_X,E_Y,E_{XY},E_{YX}$, as follows. Let $e=(u,v)\in \out(R)$, where $u\in R$, $v\not\in R$. If both $u$ and $v$ belong to $X$, then $e$ is added to $E_X$. If both vertices belong to $Y$, then $e$ is added to $E_Y$. If $u$ belongs to $X$ and $v$ to $Y$, then $e$ is added to $E_{XY}$. Otherwise, it is added to $E_{YX}$ (see Figure~\ref{fig: illustration}).
Let $E'_R=E_G(R\cap X,R\cap Y)$. If $|E_X|+|E_{XY}|\leq |E_Y|+|E_{YX}|$, then we move all vertices of $R$ to $Y$; otherwise we move them to $X$.

%\begin{figure}
%\centering
%\scalebox{0.3}{\epsfig{file=thm-cut-sparsifiers-cut.pdf}}
%\caption{Illustration for Theorem~\ref{thm: cut sparsifiers - unweighted}\label{fig: illustration}}
%\end{figure}

\begin{figure}[h]
\scalebox{0.3}{\includegraphics{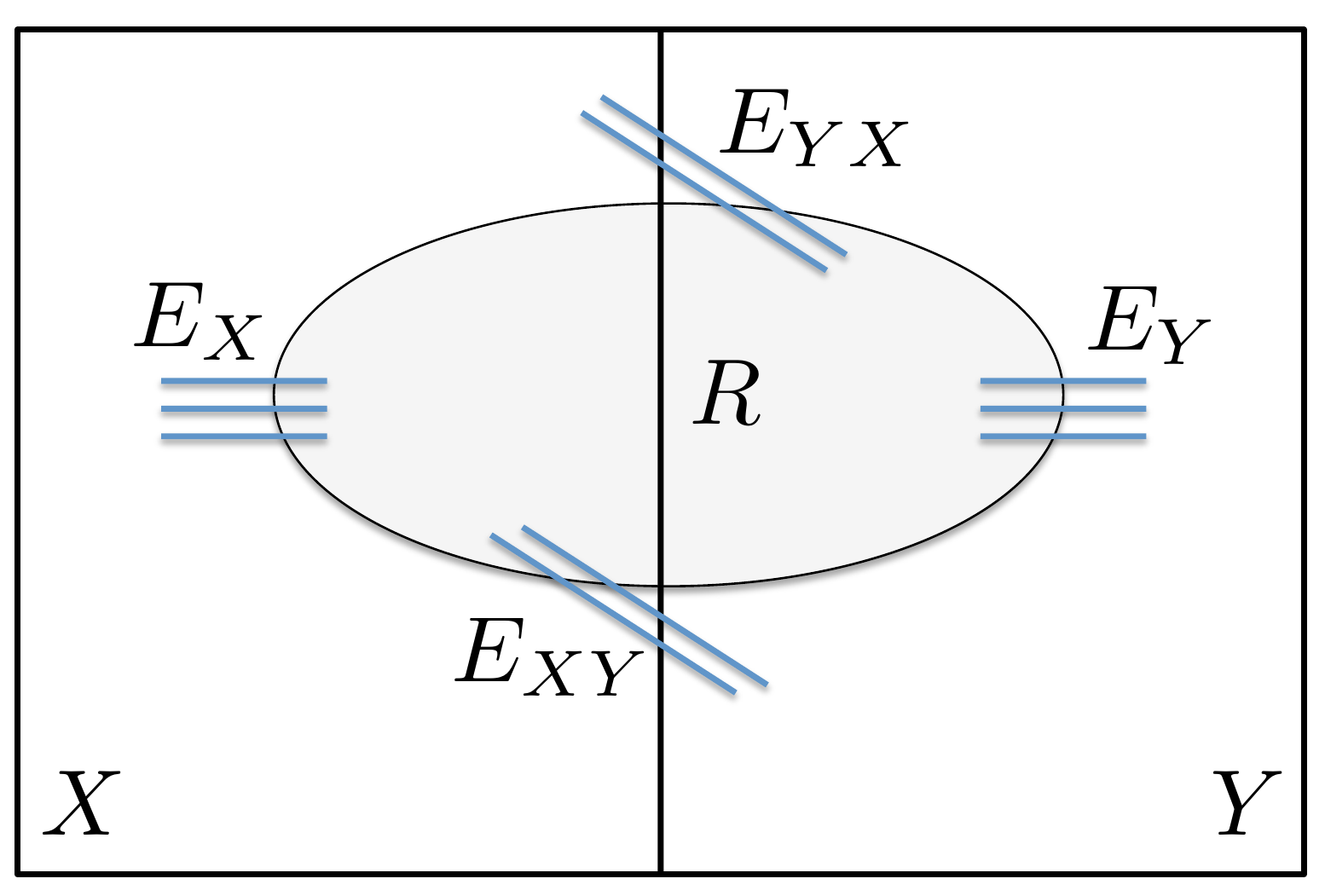}} \caption{Illustration for Theorem~\ref{thm: cut sparsifiers - unweighted}\label{fig: illustration}}
\end{figure}

Assume w.l.o.g. that $|E_X|+|E_{XY}|\leq |E_Y|+|E_{YX}|$, and so we have moved the vertices of $R$ to $Y$. The only new edges that we have added to the cut are the edges of $E_X$. On the other hand, the edges of $E'_R$, that belonged to the cut before the current iteration, do not belong to the cut anymore. We charge the edges of $E'_R$ for the edges of $E_X$. Since set $R$ is $1/3$-well-linked, $|E'_R|\geq |E_X|/3$ must hold, and so the charge to each edge of $E'_R$ is at most $3$. Moreover, since the edges of $E'_R$ are the inner edges of the set $R$ (that is, both endpoints of each such edge belong to $R$), we will never charge these edges again. Therefore, if $(\tilde{X},\tilde{Y})$ denotes the final cut, after all clusters $R\in\sset$ have been processed, then $|E_G(\tilde X,\tilde Y)|\leq 3|E_G(X,Y)|$. Finally, the cut $(\tilde X,\tilde Y)$ in graph $G$ naturally defines a cut $(X',Y')$ in graph $H$: for each cluster $R\in \sset$, if $R\sse \tilde X$, then we add $v_R$ to $X'$; otherwise we add it to $Y'$. Clearly, $|E_H(X',Y')|=|E_G(\tilde X,\tilde Y)|\leq 3|E_G(X,Y)|$. We conclude that  $\mincut_H(\tset_A,\tset_B)\leq 3\mincut_G(\tset_A,\tset_B)$.
 \end{proof}

We now complete the proof of Theorem~\ref{thm: cut sparsifiers}. Suppose we are given a graph $G$ with arbitrary edge capacities $c_e\geq 1$. For notational convenience, we denote the input parameter $\eps$ by $\eps'$, and we set $\eps=\eps'/3$. We perform the following transformation in graph $G$. Let $C$ be the sum of the capacities of all edges incident on the terminals. For each edge $e\in E$, if the capacity $c_e>C$, then we set it to be $C$. Notice that this does not change the values $\mincut_G(\tset_A,\tset_B)$ for any partition $(\tset_A,\tset_B)$ of the set $\tset$ of the terminals, since $\mincut_G(\tset_A,\tset_B)\leq C$ always holds. Finally, we replace each edge $e\in E$ with $\lceil c_e/\eps\rceil$ parallel unit-capacity edges. Let $G'$ be the resulting graph. We now apply Theorem~\ref{thm: cut sparsifiers - unweighted} to graph $G'$, to obtain a sparsifier $H'$ of size $O((C/\eps)^3)=O(C^3)$. We obtain a sparsifier $H$ for graph $G$, by setting the capacity of every edge in $H'$ to $\eps$. We now show that $H$ is a quality-$(3+3\eps)=(3+\eps')$-sparsifier for $G$.

Notice that for each partition $(\tset_A,\tset_B)$ of $\tset$, $\frac{\mincut_G(\tset_A,\tset_B)}{\eps}\leq \mincut_{G'}(\tset_A,\tset_B)\leq  \mincut_{H'}(\tset_A,\tset_B)=\frac{ \mincut_{H}(\tset_A,\tset_B)}{\eps}$, and so $\mincut_G(\tset_A,\tset_B)\leq \mincut_H(\tset_A,\tset_B)$. On the other hand, $\mincut_{G'}(\tset_A,\tset_B)\leq \frac{\mincut_G(\tset_A,\tset_B)}{\eps}(1+\eps)$, since all original edge capacities $c_e\geq 1$, and so $\lceil\frac{c_e}{\eps}\rceil\leq \frac{c_e}{\eps}(1+\eps)$. Therefore, $\mincut_{H}(\tset_A,\tset_B)=\eps\cdot \mincut_{H'}(\tset_A,\tset_B)\leq 3\eps \mincut_{G'}(\tset_A,\tset_B)\leq 3(1+\eps)\mincut_{G}(\tset_A,\tset_B)$.

%-------------------------------------------------------------------------------------------
%-------------------------------------------------------------------------------------------
%-------------------------------------------------------------------------------------------

\label{------------------------------------------Flow Sparsifiers--------------------------------------------------------}

%-------------------------------------------------------------------------------------------
%-------------------------------------------------------------------------------------------
%-------------------------------------------------------------------------------------------
\section{Flow Sparsifiers}\label{sec: flow sparsifiers}\label{SEC: FLOW SPARSIFIERS}
%-------------------------------------------------------------------------------------------
%-------------------------------------------------------------------------------------------
%-------------------------------------------------------------------------------------------

In this section we prove Theorem~\ref{thm: flow sparsifier for unit capacities} and Corollary~\ref{corollary: general flow sparsifier}. We start with the following definition.

%For convenience, we augment the graph $G$ by adding, for each terminal $t\in \tset$, a new edge $e_t=(t,t')$, so that the edges $\set{e_t\mid t\in \tset}$ are completely disjoint. Let $\tset'=\set{t'\mid t\in \tset}$, and let $G^+=(V^+,E^+)$ denote this new graph, where $V^+=V\cup \tset'$. From now on, we will think of the set $\tset'$ of vertices of $G^+$ as terminals, and we will try to find a flow sparsifier $H$ for $G^+$. Observe that if $H$ is a quality-$q$ restricted sparsifier for $G^+$, then it is also a quality-$q$ restricted sparsifier for $G$.

\begin{definition}
Let $S\sse V\setminus \tset$ be any subset of non-terminal vertices, and let $|\out(S)|=z$. We say that $S$ is a \emph{good router} iff $S$ is $1/3$-well-linked, and every pair $(e,e')\in \out(S)$ of edges can simultaneously send $1/z$ flow units to each other inside $S$, with congestion at most $\eta^*=34$. 
\end{definition}

Notice that we can efficiently check whether $S$ is a good router by computing an appropriate multicommodity flow in the graph $G_S$.
Notice also that if $S$ is a good router, then any $1$-restricted set $D$ of demands on the edges of $\out(S)$ can be routed with congestion at most $2\eta^*$ inside $S$. Indeed, let $F$ be the flow, where each pair $(e,e')\in \out(S)$ of edges sends $1/z$ flow units to each other with congestion at most $\eta^*$ inside $S$. In order to route the set $D$ of demands, consider any pair $(e,e')\in \out(S)$ of edges. Edge $e$ sends $D(e,e')/z$ flow units to each edge $e''\in \out(S)$, using the flow $F$  (scaled by factor $D(e,e')$), while edge $e'$ collects $D(e,e')/z$ flow units from each edge $e''\in \out(S)$, using the flow $F$. In the end, we have $D(e,e')$ flow units sent from $e$ to $e'$, and since the set $D$ of demands is $1$-restricted, the total congestion of this routing is bounded by $2\eta^*$.

\begin{definition}
We say that a graph $G'$ is a \emph{legal contracted graph} for $G$ iff there is a collection $\cset$ of disjoint good routers, where the clusters  $S\in \cset$ do not contain any terminals, and $G'$ is obtained from $G$ by contracting every cluster $S\in \cset$ into a super-node $v_S$. (We remove self-loops, but leave parallel edges).
\end{definition}

It is easy to see that if $G'$ is a legal contracted graph for $G$, then it is a quality-$2\eta^*$ flow sparsifier, as the next claim shows.

\begin{claim}\label{claim: a legal contracted graph is a sparsifier}
If $G'$ is a legal contracted graph for $G$, then it is a quality-$2\eta^*$ restricted flow sparsifier.
\end{claim}
\begin{proof}
Given any set $D$ of demands on the terminals in $\tset$, it is immediate to see that $\eta(G',D)\leq \eta(G,D)$, since $G'$ is obtained from $G$ by contracting some vertex subsets into super-nodes.

Assume now that we are given some set $D$ of demands on $\tset$, and $\eta(G',D)=\eta$. For simplicity, we scale the demands in $D$ down by the factor of $\eta$, to obtain a new set $D'$ of demands with $\eta(G',D')=1$. It is now enough to show that we can route the demands in $D'$ in graph $G$ with congestion at most $2\eta^*$.
 Let $F$ be the routing of $D'$ in $G'$ with congestion $1$. For each cluster $S\in \cset$, for each pair $(e,e')\in \out(S)$ of edges, let $D_S(e,e')$ be the total amount of flow in $F$ sent on flow-paths that enter $v_S$ through edge $e$, and leave it through edge $e'$. We have thus obtained a set $D_S$ of $1$-restricted demands on the edges of $\out(S)$. Since $S$ is a good router, these demands can be routed inside $S$ with congestion at most $2\eta^*$. Let $F_S$ denote this routing. In order to obtain the final routing $F'$ of the set $D'$ of demands in $G$, we start with the flow $F$, and we augment it with the routings $F_S$ that we have computed in each cluster $S\in \cset$. Therefore, $\eta(G,D')\leq 2\eta^*$.
 It is immediate to see that $G'$ is a restricted sparsifier for $G$, from the definition of a legal contracted graph. 
  \end{proof}

Most of this section is devoted to proving the following theorem, which gives a construction of a flow sparsifier for the special case where the set $V\setminus \tset$ is $1/3$-well-linked.

\begin{theorem}\label{thm: flow sparsifier: well-linked terminals}
Assume that we are given any (multi-)graph $G=(V,E)$, with unit edge capacities and a subset $\tset\sse V$ of $k$ terminals, where every vertex in $\tset$ has degree $1$. Let $R=V\setminus \tset$, and assume further that $R$ is $1/3$-well-linked. Then there is an algorithm that finds, in time $2^{k}\cdot n^{O(\log k)}$, a restricted flow sparsifier $H$ of quality $q=2\eta^*$ for $(G,\tset)$, such that $|V(H)|=k^{O(\log\log k)}$, and $H$ is a legal contracted graph for $G$. 
\end{theorem}

We defer the proof of Theorem~\ref{thm: flow sparsifier: well-linked terminals} to Section~\ref{subsec: Proof for well-linked terminals}, and complete the proof of Theorem~\ref{thm: flow sparsifier for unit capacities} here. We assume w.l.o.g that $G$ is a connected graph: otherwise, we compute a sparsifier for each of its connected components separately.
Our first step is to compute a strong well-linked decomposition $\sset$ of the set $V\setminus \tset$ of vertices, given by Theorem~\ref{thm: strong well linked}. Recall that each set $X\in \sset$ is $1/3$-well-linked, $|\sset|=O( k^3)$, and the decomposition can be computed in time $2^k\cdot \poly(n)$. For each edge $e$ in set $\bigcup_{X\in \sset}\out(X)$, we sub-divide $e$ by a new vertex $v_e$, and we let $G'$ denote the resulting graph. For each cluster $X\in \sset$, let $\tset_X=\set{v_e\mid e\in \out_G(X)}$, and let $G_X=G'[X\cup \tset_X]$. Notice that $|\tset_X|\leq k$, and all vertices in $\tset_X$ have degree $1$ in $G_X$.
For each cluster $X\in \sset$, we use Theorem~\ref{thm: flow sparsifier: well-linked terminals} on graph $G_X$ and the set $\tset_X$ of terminals, to find a restricted flow sparsifier $H_X$. Let $\cset_X$ be the corresponding collection of disjoint subsets of $V(G_X)\setminus \tset_X$, such that $H_X$ is obtained from $G_X$ by contracting every cluster in $\cset_X$. Let $\cset=\bigcup_{X\in \sset}\cset_X$. We obtain our final sparsifier $H$ by contracting every cluster $S\in \cset$ into a super-node $v_S$. Notice that since, for each cluster $X\in \sset$, graph $H_X$ is a legal contracted graph for $G_X$, each cluster $S\in \cset$ is a good router, and so $H$ is a legal contracted graph for $G$. From 
Claim~\ref{claim: a legal contracted graph is a sparsifier}, $H$ is a quality-$(2\eta^*)$ restricted sparsifier for $G$. It is easy to see that the running time of the algorithm is $2^{k}\cdot n^{O(\log k)}$. 
It now only remains to bound $|V(H)|$.

Recall that $|\sset|\leq O(k^3)$, and for each $X\in \sset$, $|\out(X)|\leq |\tset|=k$. Therefore,  $|V(H_X)|=k^{O(\log\log k)}$, and $|V(H)|=O(k^3)\cdot k^{O(\log\log k)}=k^{O(\log\log k)}$. 
This completes the proof of Theorem~\ref{thm: flow sparsifier for unit capacities}.
The proof of Corollary~\ref{corollary: general flow sparsifier} follows from Theorem~\ref{thm: flow sparsifier for unit capacities} using standard techniques, and it appears in Section~\ref{subsec: Proof of the corollary} of the Appendix.
We now focus on the proof of Theorem~\ref{thm: flow sparsifier: well-linked terminals}, which is the main technical contribution of this section.

%-----------------------------------------------------------------------------------------------------------------------------
%-----------------------------------------------------------------------------------------------------------------------------
%-----------------------------------------------------------------------------------------------------------------------------
%-----------------------------------------------------------------------------------------------------------------------------

\label{------------------------------------------Proof for well-linked terminals-----------------------------------------}
\subsection{Proof of Theorem~\ref{thm: flow sparsifier: well-linked terminals}}
\label{subsec: Proof for well-linked terminals}
\label{SUBSEC: PROOF FOR WELL-LINKED TERMINALS}

We prove the theorem by induction on the value of $k$.
Throughout the proof, we use two parameters: $r=O(\log^3k)$, and $k^*=2kr\log r=k\poly\log k$. We set the value $r$ to be a large enough integer, so that the following inequality holds:

\vspace{-5mm}

\begin{equation}
r>24\betaFCG(k^*)/\alpha_W(k^*)\label{eq: for r and kstar}
\end{equation}
\vspace{-5mm}

Notice that $\betaFCG(k^*)/\alpha_W(k^*)=O(\log^{5/2}(2kr\log r))=O(\log^3 k)+O(\log^3(r\log r))$, so $r=O(\log ^3k)$ is sufficient. 

Next, we define a function $F:\reals^+\rightarrow \reals^+$, where $F(k')$ will roughly serve as an upper bound on the size of the sparsifier for any graph $G$ with $k'$ terminals. Function $F$ is defined recursively, as follows. 
For $k'\leq 4$, $F(k')=1$. 
If $k'>4$ is an integral power of $2$, then
$F(k')=2^{16}\cdot r^3\log r\cdot F(k'/2)=O(\log^9k\log\log k)\cdot F(k'/2)$.
Otherwise, $F(k')=F(k'')$, where $k''$ is the smallest integral power of 
$2$ with $k''\geq k'$. Notice that for any integer $k'>4$, $F(k'/2)=F(\ceil{k'/2})$, and we will sometimes use these values interchangeably.

Notice that for all values $k'$, $F(k')=(\log k')^{O(\log k')}$, so $F(k)=k^{O(\log\log k)}$ as required. 
From now on, we focus on proving that if $G=(V,E)$ is a graph as in the theorem statement with $k$ terminals, then we can find, in time $n^{O(\log k)}\cdot 2^k$, a restricted  quality-$(2\eta^*)$ sparsifier $H$ for $G$, such that $|V(H)\setminus \tset|\leq F(k)$, and $H$ is a legal contracted graph for $G$. 

The proof is by induction on the values of $k$. If $k\leq 4$, then the set $R=V\setminus \tset$ is a good router, so we can let $\cset=\set{R}$, and return the corresponding contracted graph $H$ as our sparsifier, so $|V(H)\setminus\tset|=1$. Assume now that the claim holds for values $k'<k$, and we now prove it for $k$.

Notice that if the set $R=V\setminus \tset$ of vertices is a good router, then we can set $\cset=\set{R}$, and output a sparsifier $H$, obtained from $G$, after we contract the cluster $R$ into a super-node $v_R$. Therefore, we can assume from now on that $R$ is not a good router. The main idea of the algorithm is as follows. Throughout the algorithm, we maintain a collection $\cset$ of disjoint good routers in graph $G$ and the corresponding legal contracted graph $G'$. At the beginning, $\cset=\emptyset$, and $G'=G$. While the number of vertices in $V(G')\setminus \tset$ is greater than $F(k)$, we perform an iteration, in which we obtain a new collection $\cset'$ of disjoint good routers, such that the corresponding graph $G''$ contains strictly fewer vertices than $G'$. Once the number of vertices in $V(G')\setminus \tset$ falls below $F(k)$, we stop and output $G'$ as our sparsifier.
%If $k\leq 4$, then our sparsifier $H$ is simply a $4$-leaf star, and the terminals in $\tset$ are the leaves of the star. That is, we let $\cset=\set{R}$, and sparsifier $H$ is obtained from $G$ by contracting the set $R$ into a super-node $v_R$ (that becomes the center of the star). It is easy to see that $H$ is indeed a quality-$q$ sparsifier for $G$, since any set $D$ of demands on $4$ terminals can be routed in any connected graph with congestion at most $4$. We now assume the correctness of the theorem for $k'<k$, and prove it for $k$.

Notice that if $G'$ is a legal contracted graph for $G$, then each edge of $G'$ corresponds to some edge of $G$. We do not distinguish between these edges. For example, if $S\sse V(G')$ is any subset of vertices, and $S'\sse V(G)$ is obtained from $S$ by replacing each super-node $v_C\in S$ by the vertices of $C$, then we view $\out_{G'}(S)=\out_G(S')$.
We need the following definition.

\begin{definition}
Let $G'$ be the current legal contracted graph, and let $S\sse V(G')\setminus \tset$ be any subset of non-terminal vertices, such that $G'[S]$ is connected. We say that $S$ is a \emph{contractible set} iff $|\out(S)|\leq \ceil{k/2}$, and $|S|>128F(|\out_{G'}(S)|)$.
\end{definition}

Let $G'$ be the current contracted graph, and let $\cset$ be the corresponding collection of good routers. 
Suppose we can find a contractible set $S$ of vertices in the current graph $G'$, with $|\out(S)|=k'$. We show that in this case we can compute a smaller legal contracted graph $G''$. We denote this procedure by $\contract(G',S)$. Procedure $\contract(G',S)$ is executed as follows. 
Let $\cset_S\sse \cset$ contain all clusters $C$ with $v_C\in S$, and let
 $S'$ be the subset of vertices of the original graph $G$ obtained from $S$ by replacing each super-node $v_C\in \cset_S$ with the vertices of $C$. Clearly, $|S'|>128F(k')$ still holds, $G[S']$ is a connected graph, and $|\out_G(S')|=k'$. 
Let $\sset$ be the strong well-linked decomposition of $S'$ given by Theorem~\ref{thm: strong well linked}. 
We now process the clusters in $\sset$ one by one.
Consider some cluster $Z\in \sset$. We construct a new graph $G_Z$ from graph $G$, by first sub-dividing every edge $e\in \out_G(Z)$ by a vertex $v_e$, setting $\tset_Z=\set{v_e\mid e\in \out_G(Z)}$, and we let $G_Z$ be the sub-graph of the resulting graph induced by $Z\cup \tset_Z$. Let $k_Z=|\tset_Z|=|\out_G(Z)|$, and observe that $k_Z\leq k'\leq \ceil{k/2}<k$. Recall that $G_Z$ is $1/3$-well-linked for $\tset_Z$, so by the induction hypothesis, we can find a sparsifier $H_Z$ for $(G_Z,\tset_Z)$, with $|V(H_Z)\setminus \tset_Z|\leq F(k_Z)$. Let $\cset_Z$ be the collection of the good routers corresponding to $H_Z$. Recall that each cluster $C\in \cset_Z$ only contains vertices of $Z$. 
Let $\cset'=(\cset\setminus\cset_S)\cup\left (\bigcup_{Z\in \sset}\cset_Z\right )$ be the new collection of good routers in graph $G$, and let $G''$ be the contracted graph corresponding to $\cset'$. Graph $G''$ is the output of procedure $\contract(G',S)$. In the next claim we show that $|V(G'')|<|V(G')|$. 

\begin{claim}\label{claim: contraction procedure}
Let $G''$ be the output of Procedure $\contract(G',S)$. Then $|V(G'')|<|V(G')|$. 
\end{claim}

\begin{proof}
From the definition of $G''$, 

\[\begin{split}
V(G'')&= |V(G')|-|S|+\sum_{Z\in \sset}|H_Z\setminus \tset_Z|\\ & \leq  |V(G')|-|S|+\sum_{Z\in \sset}F(k_Z).\end{split}\]

Let $k'=|\out_{G'}(S)|$, and let $k''$ be the smallest power of $2$, such that $k''\geq k'$.
Recall that $|S|>128F(k')=128F(k'')$, so in order to show that $|V(G'')|<|V(G')|$, it is enough to show that $\sum_{Z\in \sset}F(k_Z)\leq 128F(k'')$. 
For each $i: 1\leq i\leq \log k''+1$, let $\sset_i\subseteq \sset$ be the collection of subsets $Z\in \sset$ with $k''/2^{i}< k_Z\leq k''/2^{i-1}$. Then from Theorem~\ref{thm: strong well linked}, $|\sset_i|\leq 2^{3i+3}$ for all $i$.
Therefore,

\[
\begin{split}
\sum_{Z\in \sset}F(k_Z)&\leq \sum_{i=1}^{\log k''+1}|\sset_i|\cdot F(k''/2^{i-1})\\
&\leq \sum_{i=1}^{\log k''+1}2^{3i+3}\cdot F(k''/2^{i-1})\end{split}\]

Let $T(i)=2^{3i-3}F(k''/2^{i-1})$. Then 
\[T(i)=8\cdot 2^{3i-6}F(k''/2^{i-1})<\half \cdot 2^{3i-6} F(k''/2^{i-2})=\half T(i-1).\]

 Therefore, values $T(i)$ form a geometrically decreasing sequence, and $\sum_{i=1}^{\log k''+1}2^{3i+3}\cdot F(k''/2^{i-1})<2^7\cdot T(1)= 2^{7}F(k'')$.
\end{proof}

We now proceed to define two structures, that we call a type-1 and a type-2 witnesses. 
We show that if $G'$ is a legal contracted graph for $G$, and $G'$ contains either a type-1 or a type-2 witness, then $R=V(G)\setminus \tset$ must be a good router. Finally, we show an algorithm, that, given a legal contracted graph $G'$ with $|V(G')\setminus \tset|>F(k)$, either finds a contractible subset $S\sse V(G')\setminus \tset$ of vertices in $G'$, or returns a type-1 or a type-2 witness in $G'$. Since we have assumed that $R$ is not a good router, whenever we apply this algorithm to the current legal contracted graph $G'$, we will obtain a contractible subset $S$ of vertices, and by using procedure $\contract(G',S)$, we can obtain a new legal contracted graph $G''$ with $|V(G'')|<|V(G')|$. We continue this process until $|V(G')\setminus \tset|\leq F(k)$ holds, and output $G'$ as our sparsifier then.
We now proceed to define the two types of witnesses.

 \begin{definition}
 Let $G'$ be a legal contracted graph, and let $\fset=\set{S_1',\ldots,S_r'}$ be a family of disjoint subsets of $V(G')\setminus \tset$. We say that $\fset$ is a type-1 witness, iff for each $1\leq j\leq r$, $S_j'$ is $\alpha_W(k^*)$-well-linked in graph $G'$, and there is a collection $\pset_j'$ of $\ceil{k/2}$ edge-disjoint paths in graph $G'$, where each path connects a distinct terminal in $\tset$ to a distinct edge in $\out_{G'}(S_j')$.
\end{definition}

\begin{definition}
Let $\tA\sse V(G')\setminus \tset$ be any subset of non-terminal vertices. We say that $\tA$ is a type-2 witness iff we are given a subset $\tE\sse \out_{G'}(\tA)$ of $r\cdot \ceil{k/4}$ edges, such that $\tA$ is $\alpha_W(r\cdot\ceil{k/4})$-well linked for $\tE$, and we are given a partition $E_1,\ldots,E_r$ of $\tE$ into $r$ disjoint subsets of size $\ceil{k/4}$ each, and a subset $\tset^*\sse \tset$ of $\ceil{k/4}$ terminals, such that for each $1\leq j\leq r$, there  is a collection $\pset_j':\tset^*\sconnect_2E_j$ of paths in graph $G'$.
\end{definition}

(Here we say that $\tA$ is $\alpha$-well-linked for $\tE\sse \out_{G'}(\tA)$ iff for any partition $(X,Y)$ of $\tA$, if we denote $T_X=\tE\cap \out_{G'}(X)$, and $T_Y=\tE\cap \out_{G'}(Y)$, then $|E_{G'}(X,Y)|\geq \alpha\cdot \min\set{|T_X|,|T_Y|}$.)

We start by showing that if a legal contracted graph $G'$ contains a type-1 witness or a type-2 witness, then the set $R$ is  good router.

%-----------------------------------------------------------------------
%-----------------------------------------------------------------------
\begin{theorem}\label{thm: from witnesses to good routers}
If any legal contracted graph $G'$ contains a type-1 witness $\fset$, or a type-2 witness $\tA$, then $R=V(G)\setminus \tset$ is a good router.
\end{theorem}
%-----------------------------------------------------------------------
%-----------------------------------------------------------------------

\begin{proof}
Recall that $R$ is $1/3$-well-linked. So we only need to prove that if $G'$ contains a type-1 or a type-2 witness, then every pair of terminals can simultaneously send $1/k$ flow units to each other with congestion at most $\eta^*$. We need the following two simple claims, whose proofs appear in the Appendix.

%-------------------------------------------------------------------------------------------
%-------------------------------------------------------------------------------------------
%-------------------------------------------------------------------------------------------

%-------------------------------------------------------------------------------------------
\begin{claim}\label{claim: well-linkedness of sets in G}
Let $G'$ be a legal contracted graph, $S'\sse V(G')\setminus \tset$, and $E'\sse \out_{G'}(S')$, such that $S'$ is $\alpha$-well-linked for $E'$, for any $\alpha<1$. Let $S\sse V(G)\setminus\tset$ be the set of vertices obtained from $S'$, after we replace every super-node $v_C\in S'$ with the set $C$ of vertices. Then $S$ is $\alpha/3$-well-linked for $E'$ in graph $G$.
\end{claim}

%------------------------------------------------------------------------------------------

\begin{claim}\label{claim: paths in contracted and original graphs}
Let $G'$ be a legal contracted graph for $G$, $S'\sse V(G')\setminus \tset$ any subset of non-terminal vertices in $G'$, and $E'\sse \out_{G'}(S')$ any subset of edges, and assume further that we are given a subset $\tset'\sse \tset$ of terminals with $|\tset'|=|E'|$, such that there is a collection $\pset':\tset'\sconnect_{\eta}E'$ of paths in $G'$. Let $S\sse V(G)\setminus \tset$ be the set of vertices obtained from $S'$ after we replace every super-node $v_C$ by the set $C$ of vertices, and consider the same subset $E'\sse \out_{G}(S)$ of edges. Then there is a set $\pset:\tset'\sconnect_{3\eta}E'$ of paths in graph $G$.
\end{claim}

\paragraph{Type-1 Witnesses}
Assume first that graph $G'$ contains a type-1 witness $\fset=\set{S_1',\ldots,S_r'}$. Fix some $1\leq j\leq r$, and consider the subset $S_j'$ of vertices. Let $S_j$ be the corresponding subset of vertices of the original graph $G$, after we un-contract each super-node $v_C\in S_j$, replacing it with the corresponding set $C$ of vertices. 
Let $\tset_j\sse \tset$ be the subset of $\ceil {k/2}$ terminals that serve as endpoints of the paths in $\pset_j'$, and let $E_j\sse \out_{G'}(S_j')$ be the subset of $\ceil{k/2}$ edges where these paths terminate. 
From Claim~\ref{claim: well-linkedness of sets in G}, set $S_j$ is $\alpha_W(k^*)/3$-well-linked. 
Therefore, every pair $(e,e')\in E_j$ of edges can simultaneously send to each other at least  $\frac{1}{\ceil{k/2}}\cdot\frac{\alpha_W(k^*)}{3\betaFCG(k^*)}\geq \frac{8}{kr}$ flow units with no congestion in $G[S_j]$. (We have used Equation~\ref{eq: for r and kstar}). Denote this flow by $F_j$.
From Claim~\ref{claim: paths in contracted and original graphs}, there is a set $\pset_j: \tset_j\sconnect_3 E_j$ of paths in graph $G$. Let $\tset'_j=\tset\setminus\tset_j$. Then $|\tset_j'|\leq |\tset_j|\leq k/2$. Since graph $G$ is $1/3$-well-linked for $\tset$, there is a set $\pset_j^*: \tset'_j\sconnect_3\tset_j$ of paths in graph $G$. We now define a flow $F^*_j$, as follows: each terminal $t\in \tset_j'$ sends $1/r$ flow units to some terminal in $\tset_j$, along the path in $\pset^*_j$ that originates at $t$. Next, each terminal $t'\in \tset_j$ sends $\frac{2-1/k} r$ flow units to some edge in $E_j$, using the path in $\pset_j$ that originates at $t'$. 
Each edge in $E_j$ now receives $\frac{2-1/k} r$ flow units, and uses the flow $F_j$ to spread this flow evenly among the edges of $E_j$. This defines the flow $F_j^*$, where every pair $(t,t')$ of terminals sends $\frac{1}{kr}$ flow units to each other. The congestion of the flow $F_j^*$ is computed as follows: the congestion due to flow on paths in $\pset_j^*$ is at most $3/r$; the congestion due to flow on paths in $\pset_j$ is at most $6/r$, and the congestion due to the flow $F_j$ is at most $1$. Notice that flow $F_j$ is entirely contained inside $G[S_j]$.

The final flow $F^*$ is simply the union of flows $F_j$ for $1\leq j\leq r$. Clearly, in $F^*$, every pair of terminals sends $1/k$ flow units to each other. It is easy to see that the flow congestion is bounded by $10$.

\paragraph{Type-2 Witnesses}
Assume now that we are given a type-2 witness $\tA$, and let $A\sse V(G)\setminus\tset$ be the subset of vertices obtained from $\tA$, after we replace each super-node $v_C$ with the set $C$ of vertices. From Claim~\ref{claim: well-linkedness of sets in G}, set $A$ is $\frac 1 3 \alpha_W(r\ceil{k/4})$-well-linked for the subset $\tE\sse \out_G(A)$ of edges. Therefore, every pair $(e,e')\in \tE$ of edges can send $\frac{1}{r\cdot \ceil{k/4}}\cdot \frac{\alpha_W(r\cdot \ceil{k/4})}{3\betaFCG(r\cdot \ceil{k/4})}>\frac{16}{kr^2}$ flow units to each other with no congestion in graph $G$.  (We have used Equation~\ref{eq: for r and kstar} and the fact that $k^*>r\cdot \ceil{k/4}$). Let $F$ denote this flow. Recall that for each $1\leq j\leq r$, we have a collection $\pset'_j: \tset^*\sconnect_2E_j$ of paths in graph $G'$. From Claim~\ref{claim: paths in contracted and original graphs}, there is a set $\pset_j: \tset^*\sconnect_6 E_j$ of paths in graph $G$.
Finally, partition $\tset\setminus\tset^*$ into three subsets, $\tset_1,\tset_2,\tset_3$ of size at most $\ceil{k/4}$ each. Since graph $G$ is $1/3$-well-linked, for each $1\leq i\leq 3$, there is a set $\qset_i:\tset_i\sconnect_3\tset^*$ of paths in $G$. Let $\qset$ denote the following set of paths: start with $\qset_1\cup \qset_2\cup \qset_3$, and add, for each terminal $t\in \tset^*$, an empty path $Q_t$ connecting $t$ to itself. Then set $\qset$ contains, for each terminal $t\in \tset$, a path $Q_t$, connecting $t$ to some terminal $t'\in \tset^*$, such that for each terminal $t'\in \tset^*$, there are exactly four terminals in $\tset$ whose path $Q_t$ terminates at $t'$. 
%For each terminal $t'\in \tset^*$, let $\phi(t')$ denote the set of these four terminals. 
Notice that the paths in $\qset$ cause congestion at most $9$ in $G$. We are now ready to define our final flow $F^*$. First, every terminal in $\tset$ sends one flow unit to some terminal in $\tset^*$, along the path $Q_t\in \qset$. Next, for each $1\leq j\leq r$, each terminal $t\in \tset^*$, sends $\frac{4} r$ flow units along the path in $\pset_j$ that originates at $t$. Notice that each edge in $\tE$ now receives $\frac{4} r$ flow units. Finally, we use the flow $F$, to spread the flow that every edge receives evenly among the edges in $\tE$, so every pair of edges in $\tE$ needs to send $\frac{4}{r}\cdot \frac{1}{r\cdot \ceil{k/4}}\leq \frac{16}{kr^2}$ flow units to each other. This finishes the definition of the flow $F^*$. Clearly, every pair of terminals sends $1/k$ flow units to each other. We now analyze the congestion due to this flow. The paths in $\qset$ cause congestion $9$, and the paths in $\pset_1,\ldots,\pset_r$ cause congestion at most $24$ altogether 
(each set $\pset_j$ of paths originally caused congestion $6$, and we send $4/r$ flow units along each path in $\pset_j$). Finally, flow $F$ causes congestion at most $1$. Altogether, flow $F^*$ causes congestion at most $34$.
 \end{proof}

The next theorem provides an algorithm that, given any legal contracted graph $G'$, either finds a contractible subset of vertices in $G'$, or finds a witness of type $1$ or $2$ in $G'$.
%------------------------------------------------------------
%------------------------------------------------------------
%------------------------------------------------------------
%------------------------------------------------------------
\begin{theorem}\label{thm: can find good structures or contract further}
Let $G'$ be any legal contracted graph, and assume that $|V(G')\setminus \tset|>F(k)$. Then there is an efficient algorithm that finds either a contractible subset $S'$ of vertices, or a type-1 witness $\fset$, or a type-2 witness $\tA$ in graph $G'$.
\end{theorem}
%------------------------------------------------------------
%------------------------------------------------------------
%------------------------------------------------------------

\begin{proof}
Since we only work with graph $G'$ in this proof, we omit the sub-script $G'$ in our notation, and use $\out(S)$ to denote $\out_{G'}(S)$.
Let $S\sse V(G')\setminus \tset$ be any subset of non-terminal vertices. We say that a partition $(X,Y)$ of $S$ is \emph{balanced}, iff $|X|,|Y|\geq |S|/4$. We start with the following lemma.

\begin{lemma} \label{lemma: balanced-cut}
Let $S\sse V(G')\setminus \tset$ be any subset of non-terminal vertices with $|S|>2^9\cdot F(k/2)$. Then there is an efficient algorithms that either finds a type-2 witness $\tilde A$, or a contractible set $S'$ of vertices in $G'$, or a balanced partition $(X,Y)$ of $S$ with $|E(X,Y)|\leq rk$.
\end{lemma}

\begin{proof}
Let $(X,Y)$ be any balanced partition of $S$, and assume w.l.o.g. that $|X|\geq |Y|$. If $|E(X,Y)|\leq rk$, then we stop and output the partition $(X,Y)$. Otherwise, we perform a number of iterations. In each iteration, we are given as input a balanced partition $(X,Y)$ of $S$ with $|X|\geq |Y|$ and $|E(X,Y)|> rk$, and we
try to establish whether $X$ is a type-2 witness. If this is not the case, then we will either find a contractible subset 
$S'$ of vertices in $G'$, or we will produce a new balanced partition $(X',Y')$ of $S$, with $|E(X',Y')|<|E(X,Y)|$. 
Therefore, after at most $|E(G')|$ steps, we are guaranteed to find a type-2 witness $\tilde A$, or a contractible set $S'$ of vertices, or a balanced partition $(X,Y)$ of $S$ with $|E(X,Y)|\leq rk$.

We now proceed to describe each iteration. Suppose we are given a balanced partition $(X,Y)$ of $S$ with $|X|\geq |Y|$ and $|E(X,Y)|> rk$. Throughout the iteration execution, we denote $\Gamma=E(X,Y)$. An iteration consists of three steps. In the first step, we try to find a collection $\pset_1$ of $\ceil{k/4}$ edge-disjoint paths in graph $G'$ connecting $\ceil{k/4}$ distinct terminals in $\tset$ to a subset $E_1$ of $\ceil{k/4}$ edges in $\Gamma$. In the second step, we identify additional $(r-1)$ subsets $E_2,\ldots,E_r$ of edges of $\Gamma$ of size $\ceil{k/4}$ each, and try to find, for each $1\leq j\leq r$, a collection $\pset_j$ of paths connecting terminals in $\tset$ to the edges in $E_j$ with congestion at most $2$. Finally, in the third step, we set $\tE=\bigcup_{j=1}^rE_j$, and we try to establish whether $X$ is $\alpha_W(r\cdot \ceil{k/4})$-well-linked for $\tE$. If all three steps succeed, then we output $X$ as a type-2 witness. If any of the three steps fails, then we will either find a contractible set $S'$ of vertices in $G'$, or a new balanced partition $(X',Y')$ of $S$ with $|E(X',Y')|<|E(X,Y)|$. In the latter case, we continue to the next iteration with the new partition $(X',Y')$ replacing the partition $(X,Y)$. We now turn to describe each of the three steps.

\paragraph{Step 1}
In this step we try to find a set $\pset_1$ of edge-disjoint paths in graph $G'$ connecting 
$\ceil{k/4}$ distinct  terminals in $\tset$ to the edges of $\Gamma$. 
In order to do so, we set up the following flow network $N$. We sub-divide each edge $e\in \Gamma$ by a vertex $z_e$, and set $\tset'=\set{z_e\mid e\in \Gamma}$. We then contract the vertices of $\tset$ into a source $s$, and the vertices of $\tset'$ into a sink $t$. Assume first that there is an $s$-$t$ flow of value at least $\ceil{k/2}$ in the resulting network $N$. This flow defines a collection $\pset'$ of $\ceil{k/2}$ paths, where each path connects a distinct terminal in $\tset$ to some edge in $\Gamma$ (since each terminal in $\tset$ has exactly one adjacent edge in $G'$). These paths are completely edge-disjoint, except that each edge in $\Gamma$ may serve as an endpoint of up to two such paths. We  select a subset $\pset_1\sse \pset'$ of $\ceil{k/4}$ paths, such that each edge in $\Gamma$ now participates in at most one path in $\pset_1$, that is, the paths in $\pset_1$ are edge-disjoint.

Assume now that the value of the maximum $s$-$t$ flow in $N$ is less than $\ceil{k/2}$. We show that in this case, we can either find a contractible set $S'$ of vertices, or a balanced partition $(X',Y')$ of $S$ with $|E(X',Y')|\leq \ceil{k/2}<rk$. Since the value of the maximum $s$-$t$ flow in $N$ is less than $\ceil{k/2}$, there is an $s$-$t$ cut $(A',B')$ with $s\in A'$, $t\in B'$, and $|E(A',B')|<\ceil{k/2}$ in $N$. Let $A=A'\setminus \set{s}$ and $B=B'\setminus\set{t}$. Then $(A,B)$ is a partition of  $V(G')\setminus \tset$. Denote $X_A=X\cap A$, $X_B=X\cap B$, $Y_A=Y\cap A$, and $Y_B=Y\cap B$. Let $\Gamma'\sse \Gamma$ be the subset of edges $e=(u,v)$ where either $u\in A$, $v\in B$, or both $u,v\in A$ (see Figure~\ref{fig: for lemma-balanced-cut}). Notice that every edge in $\Gamma'$ contributes at least $1$ to the cut $E_N(A',B')$, and since $E_{G'}(A,B)\sse E_N(A,B)\cup \Gamma'$, we get that $|E_{G'}(A,B)|<\ceil{k/2}$.

\begin{figure}[h]
\scalebox{0.3}{\rotatebox{0}{\includegraphics{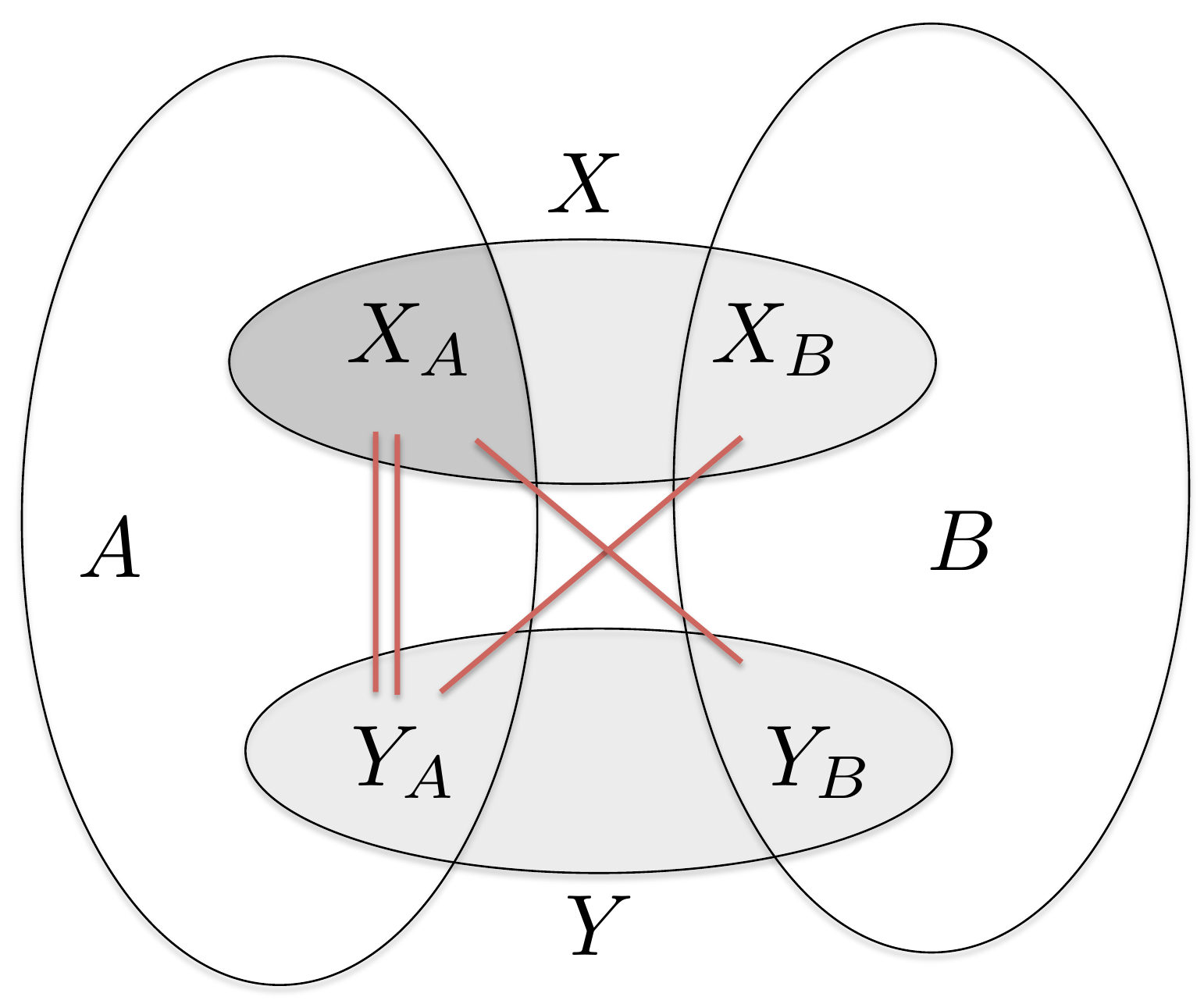}}}  \caption{\label{fig: for lemma-balanced-cut}Illustration for Lemma~\ref{lemma: balanced-cut}. Edges in $\Gamma'$ are shown in red.}
\end{figure}

Assume first that $|X_A|\geq |X_B|$. In this case, we define a new partition $(X',Y')$ of $S$, where 
$X'=X_A$ and $Y'=Y\cup X_B$. It is immediate to see that $(X',Y')$ is a balanced cut, since $|X'|\geq |X|/2\geq |S|/4$. 
In order to bound $E(X',Y')$, observe that 

\[
E_{G'}(X',Y')=E_{G'}(X_A,X_B)\cup E_{G'}(X_A,Y)\subseteq E_{G'}(X_A,X_B)\cup \Gamma'.\]

Therefore, $|E_{G'}(X',Y')|\leq |E_{G'}(X_A,X_B)|+|\Gamma'|<\ceil{k/2}<kr$.

From now on we assume that $|X_B|\geq |X_A|$, so $|X_B|\geq |S|/4>2^7F(k/2)$. Let $\cset_1$ be the set of all connected components of $G'[B]$. If for any component $C\in \cset_1$, $|C|>2^7F(k/2)$, then we stop the algorithm, and output $C$ as a contractible set. Indeed, $|C|>2^7F(k/2)$, while $|\out(C)|\leq |E_G'(A,B)|<\ceil{k/2}$. We now assume that for all components $C\in \cset_1$, $|C|\leq 2^7F(k/2)<|S|/4$. 

Let $\cset_2$ be the set of all connected components of $G'[X_B\cup Y_B]$. Notice that each connected component $C\in \cset_2$ must be contained in some connected component $C'\in \cset_1$, so $|C|<|S|/4$ must hold.
We construct a new partition $(X',Y')$ of $S$, as follows. Start with $X'=\emptyset$, and 
add components $C\in \cset_2$ to $X'$ one-by-one, until $|X'|\geq |S|/4$ holds. Since the size of each such component is less than $|S|/4$, while $|X_B|\geq |S|/4$, in the end, $|S|/4\leq |X'|\leq |S|/2$. Let $Y'=S\setminus X'$. Then $(X',Y')$ is a balanced partition of $S$, and $E_{G'}(X',Y')|\sse E_{G'}(A,B)$, so $|E_{G'}(X',Y')|\leq |E_{G'}(A,B)|\leq \ceil{k/2}<kr$.

\paragraph{Step 2}
From now on, we assume that we have successfully found a set $\pset_1$ of $\ceil{k/4}$ edge-disjoint paths connecting a subset $\tset^*\sse \tset$ of $\ceil{k/4}$ terminals to the edges in $\Gamma$. 
 Let $\Gamma_1$ be the subset of $\ceil{k/4}$ edges of $\Gamma$ that serve as endpoints of these paths, so $\pset_1: \tset^*\sconnect_1\Gamma_1$.

We select arbitrary $(r-1)$ disjoint subsets $\Gamma_2,\ldots,\Gamma_r$ of $\Gamma\setminus \Gamma_1$ , containing $\ceil{k/4}$ edges of each. For each $2\leq j\leq \Gamma_j$, we will try to find a collection $\pset_j'$ of edge-disjoint paths, connecting the edges of $\Gamma_1$ to the edges of $\Gamma_j$. We will show that if such set of paths cannot be found, then we can find another balanced partition $(X',Y')$ of $S$ with $|E_{G'}(X',Y')|<|E_{G'}(X,Y)|$. For simplicity, we provide and analyze the procedure for $j=2$, and the procedure is similar for all $2\leq j\leq r$.

We set up the following flow network. Start with the graph $G'[X]\cup \Gamma_1\cup\Gamma_2$. Let $V_1$ be the set of the endpoints of edges of $\Gamma_1$ that do not belong to $X$, $V_1=\set{v\mid (v,u)\in \Gamma_1,v\not\in X}$, and we define $V_2$ similarly for $\Gamma_2$. We then unify all vertices of $V_1$ into a source $s$, and all vertices of $V_2$ into a sink $t$. Let $N'$ be the resulting network. Assume first that there is an $s$-$t$ flow in $N'$ of value $\ceil{k/4}$. Then this flow defines a collection $\pset_2'$ of $\ceil{k/4}$ edge-disjoint paths, connecting the edges of $\Gamma_1$ to the edges of $\Gamma_2$. Concatenating the paths in $\pset_1$ with the paths in $\pset_2'$, we obtain a collection $\pset_2: \tset^*\sconnect_2\Gamma_2$ of paths in $G'$.

Assume now that such flow does not exist. Then there is an $s$-$t$ cut $(A,B)$ in $N'$, with $s\in A$, $t\in B$, and $|E(A,B)|<\ceil{k/4}$.
We partition the edges of $\Gamma_1$ into two subsets: set $T_A$ denotes the edges that do not belong to the cut $E_{N'}(A,B)$ (that is, for each edge $e=(s,v)\in T_A$, $v\in A$), and set $T_B$ denotes edges that belong to the cut (for each edge $e=(s,v)\in T_B$, $v\in B$). Similarly, we partition the set $\Gamma_2$ of edges as follows: set $T'_A$ contains all edges that belong to the cut $E_{N'}(A,B)$, and $T'_B$ contains all edges that do not belong to the cut. The set $\Gamma\setminus (\Gamma_1\cup \Gamma_2)$ of edges is also partitioned into two subsets: $\Upsilon_A$ denotes all edges $(u,v)\in \Gamma$ with $u\in Y$ and $v\in A$, and $\Upsilon_B$ denotes all edges  $(u,v)\in \Gamma$ with $u\in Y$ and $v\in B$. Finally, let $E'=E_{G'}(A,B)$
(See Figure~\ref{fig: cut claim}).

\begin{figure}[h]
\scalebox{0.3}{\rotatebox{0}{\includegraphics{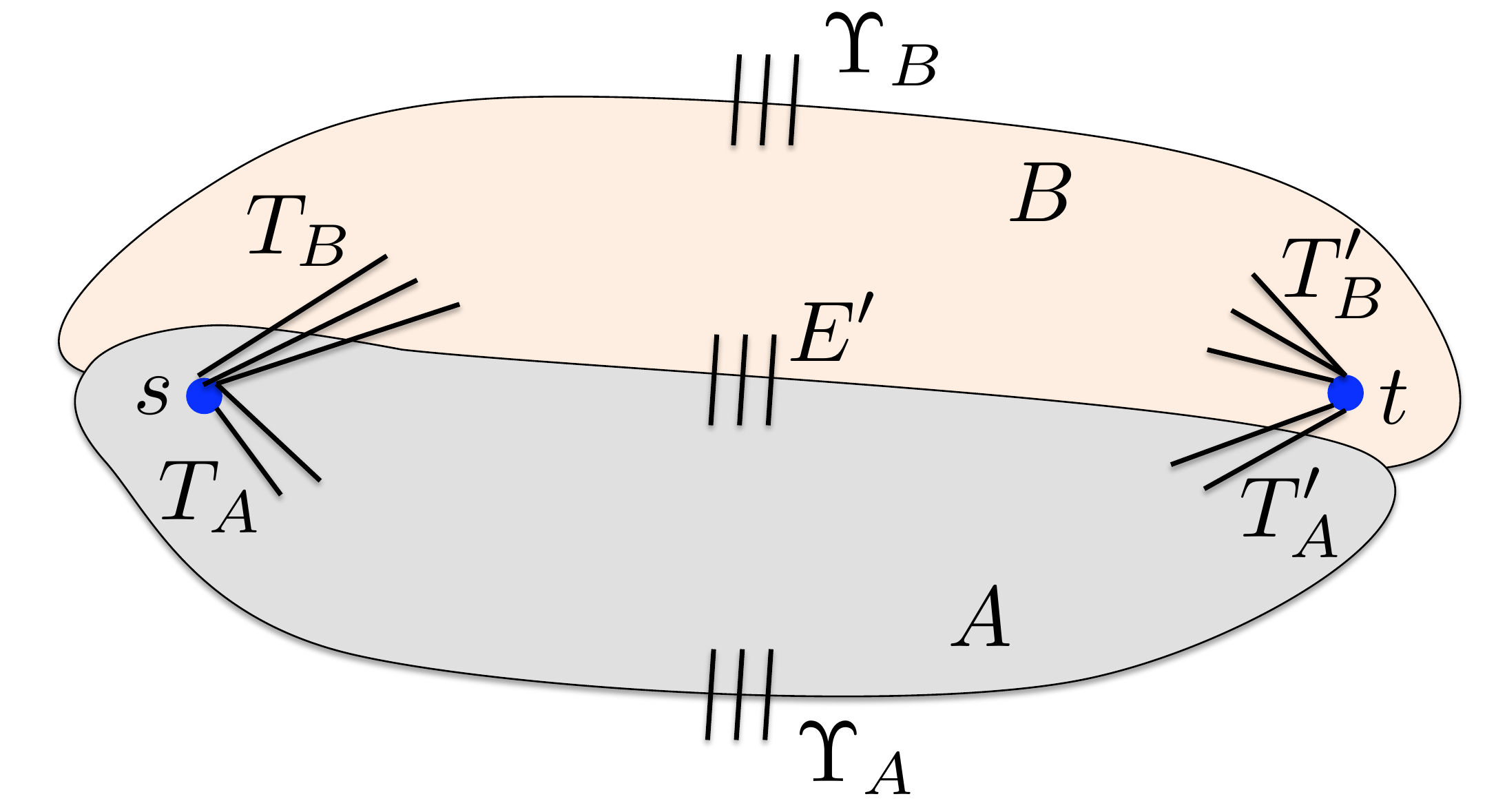}}}  \caption{\label{fig: cut claim}Illustration for Lemma~\ref{lemma: balanced-cut}}
\end{figure}

The set of edges that belong to the cut $E_{N'}(A,B)$ is $E'\cup T_B\cup T'_A$, and the value of this cut is less than $\ceil{k/4}$. In particular, since $|T_A\cup T_B|=\ceil{k/4}$ and $|T_A'\cup T_B'|=\ceil{k/4}$, it follows that $|E'|< |T_A|$, and $|E'|<|T_B'|$.
Assume first that $|A|\leq |B|$. We then define a new partition $(X',Y')$ of $S$, where $X'=B$ and $Y'=Y\cup A$. It is easy to see that $(X',Y')$ is a balanced cut. Notice that $E_{G'}(X,Y)=T_A\cup T'_A\cup T_B\cup T'_B\cup \Upsilon_A\cup \Upsilon_B$, while $E(X',Y')=T_B\cup T_B'\cup \Upsilon_B\cup E'$. In order to show that $|E(X',Y')|<|E(X,Y)|$, it is enough to prove that $|E'|<|T_A|+|T_A'|$, which follows from the fact that $|E'|<|T_A|$.

Otherwise, if $|A|>|B|$, we define a new partition $(X',Y')$ of $S$ where $X'=A$ and $Y'=Y\cup B$. Again, it is easy to see that $(X',Y')$ is a balanced cut.  Notice that $E(X,Y)=T_A\cup T'_A\cup T_B\cup T'_B\cup \Upsilon_A\cup \Upsilon_B$, while $E(X',Y')=T_A\cup T_A'\cup \Upsilon_A\cup E'$. In order to show that $|E(X',Y')|<|E(X,Y)|$, it is enough to prove that $|E'|<|T_B|+|T_B'|$, which follows from the fact that $|E'|<|T_A|$.

We say that steps 1 and 2 are successful iff we have found $r$ disjoint subsets $\Gamma_1,\ldots,\Gamma_r$ of $\Gamma$ containing $\ceil{k/4}$ edges each, and for each $1\leq j\leq r$, we have found a set $\pset_j$ of $\ceil{k/4}$ edge-disjoint paths, $\pset_j: \tset^*\sconnect_2 \Gamma_j$. We assume from now on that steps 1 and 2 have been successful. We now proceed to describe step 3.

\paragraph{Step 3}
Let $\Gamma'=\bigcup_{j=1}^r\Gamma_j$.
In this step, we try to verify that $X$ is $\alpha_W(r\cdot \ceil{k/4})$-well-linked for $\Gamma'$. If this is not the case, then we return a balanced partition $(X',Y')$ of $S$ with $|E_{G'}(X',Y')|<|E_{G'}(X,Y)|$.
We set up an instance of the sparsest cut problem, as follows. Start with the graph $G'$ and sub-divide every edge $e\in \Gamma'$ by a vertex $v_e$. Let $\tset'=\set{v_e\mid e\in \Gamma'}$, and let $G''$ be the sub-graph of the resulting graph  induced by $X\cup \tset'$. We run algorithm \algsc on the instance $(G'',\tset')$ of the sparsest cut problem. Let $A,B$ be the resulting partition of $X$, and assume w.l.o.g. that $|A|\leq |B|$. Denote $T_A=\out_{G'}(A)\cap \Gamma'$ and $T_B=\out_{G'}(B)\cap \Gamma'$.
Assume first that $|E_{G'}(A,B)|<\min\set{|T_A|,|T_B|}$. We then define a new partition $(X',Y')$ of $S$, where $X'=B$ and $Y'=A\cup Y$. It is easy to see that $(X',Y')$ is a balanced partition, since $|B|\geq |A|$. Moreover, $|E_{G'}(X',Y')|\leq |E_{G'}(X,Y)|-|T_A|+|E_{G'}(A,B)|< |E_{G'}(X,Y)|$ as required.

Assume now that $|E_{G'}(A,B)|\geq \min\set{|T_A|,|T_B|}$. Then we are guaranteed that $X$ is $(1/\alphasc(r\cdot \ceil{k/4}))\geq \alpha_W(r\cdot \ceil{k/4})$-well-linked for $\Gamma'$. We then declare that $X$ is a type-2 witness and terminate the algorithm. Indeed, we have established that $X$ is $\alpha_W(r\cdot \ceil{k/4})$-well-linked for $\Gamma'$, and we have found, for each $1\leq j\leq r$, a collection $\pset_j:\tset^*\sconnect_2\Gamma_j$ of paths in $G'$.
 \end{proof}

We are now ready to complete the proof of Theorem~\ref{thm: can find good structures or contract further}. 
The algorithm consists of two phases. In the first phase we have $\lceil \log r\rceil$ iterations. In each iteration $i$, we start with a family $\sset_i$ of $2^{i-1}$ disjoint subsets of vertices of $V(G')\setminus \tset$, where for each $S\in \sset_i$ $|S|> 2^9F(k/2)$, and produce a family $\sset_{i+1}$ of $2^i$ subsets, that become an input to the next iteration. In the input to the first iteration, $\sset_1=\set{V(G')\setminus \tset}$. Iteration $i$ is executed as follows. Consider some set $S\in \sset_i$. We apply the algorithm from Lemma~\ref{lemma: balanced-cut} to set $S$. If the output is a type-2 witness $\tA$, or a contractible set $S'$ of vertices, we stop the algorithm and output this set. Otherwise, we obtain a balanced partition $(X,Y)$ of $S$, with $|E(X,Y)|\leq rk$. In this case, we add $X$ and $Y$ to $\sset_{i+1}$. Notice that $|\out(X)|,|\out(Y)|\leq |\out(S)|+rk$. We let $\sset_{i+1}$ be the set obtained after we process all sets $S\in \sset_i$.
Observe that since we find balanced cuts in each iteration, for each $1\leq i\leq \lceil\log r\rceil+1$, for each $S\in \sset_i$, $|S|\geq \frac{|V(G')\setminus \tset|}{4^{i-1}}\geq \frac{F(k)}{4^{\lceil \log r\rceil}}\geq \frac{F(k)}{4r^2}> 2^9\cdot F(k/2)$, and so we can indeed apply Lemma~\ref{lemma: balanced-cut} to all sets in $\sset_i$.

Consider now the output of the last iteration $\sset_{\lceil\log r\rceil}$, and let $S_1,\ldots,S_r$ be any $r$ sets in $\sset_{\lceil\log r\rceil+1 }$. Fix some $j: 1\leq j\leq r$, and consider the set $S_j$. From the above discussion, $S_j\sse V(G')\setminus \tset$. Moreover, since $|\out(V(G')\setminus \tset)|=k$, and in each iteration, if we start with a set $S$ and produce a partition $(X,Y)$ of $S$, then $|\out(X)|,|\out(Y)|\leq |\out(S)|+rk$, we get that $|\out(S_j)|\leq k+rk\cdot\lceil \log r\rceil\leq 2kr\log r=k^*$.
Moreover, as observed above, $|S_j|\geq \frac{F(k)}{16r^2}\geq 2^{12}r\log r\cdot F(k/2)$. 

Let $\wset_j$ be the weak well-linked decomposition of $S_j$, given by Theorem~\ref{thm: weak well-linked}. Notice that from the definition of well-linkedness, for every cluster $C\in \wset_j$, $G'[C]$ is connected. If any set $R'\in \wset_j$, with $|\out(R')|\leq \ceil{k/2}$ is contractible, then we simply output $R'$ as a contractible set. From now on assume that all sets in $\wset_j$ are non-contractible. Notice that for each $R'\in \wset_j$, set $R'$ is $\alpha_W(k^*)$-well-linked. Let $S'_j\in \wset_j$ be the set of maximum cardinality. We need the following claim.

\begin{claim} $|S_j'|>2^7\cdot F(k/2)$.
\end{claim}
\begin{proof}
Recall that from Theorem~\ref{thm: weak well-linked}, $\sum_{R'\in \wset_j}|\out(R')|\leq 1.2|\out(S_j)|\leq 2.4kr\log r$.
We partition the set $\wset_j$ into two subsets: $\wset^1$ contains all sets $R'\in \wset_j$ with $|\out(R')|\geq k/2$, and $\wset^2$ contains all remaining sets. Further, we partition the set $\wset^2$ into subsets $\rset_i$, for
 $2\leq i\leq \log k +1$, as follows: $\rset_i$ contains all sets $R'$ with $k/2^i<|\out(R')|\leq k/2^{i-1}$. 
 
 Fix some $2\leq i\leq \log k+1$. Since  $\sum_{R'\in \wset_j}|\out(R')|\leq 2.4kr\log r$, we get that $|\rset_i|\leq 3r\log r\cdot 2^i$, and since each set $R'\in \rset_i$ is non-contractible, $|R'|\leq 2^7\cdot F(k/2^{i-1})$ must hold. We therefore obtain the following bound:
 
 \[\begin{split}
 \sum_{R'\in \wset^2}|R'|&\leq \sum_{i=2}^{\log k +1}3\cdot 2^{i}\cdot r\log r\cdot 2^7\cdot F(k/2^{i-1})
\\&=3\cdot 2^7\cdot r\log r\sum_{i=2}^{\log k +1}2^iF(k/2^{i-1})
 \end{split}\]

 Denote $T(i)=2^iF(k/2^{i-1})$. By the recursive definition of $F(k')$, $T(i)<\frac{2^i}{8}F(k/2^{i-2})=T(i-1)/4$. Therefore, the values $T(i)$ form a geometric series, and $\sum_{i=2}^{\log k +1}T(i)<4T(2)/3$.
We conclude that $\sum_{R'\in \wset^2}|R'|< 2^{11}r\log r\cdot F(k/2)\leq |S_j|/2$, and so $ \sum_{R'\in \wset^1}|R'|> |S_j|/2\geq 2^{11}r\log r\cdot F(k/2)$.

Finally, observe that set $\wset^1$ may contain at most $5r\log r$ sub-sets, since for each subset $R'\in \wset^1$, $|\out(R')|>k/2$, while $\sum_{R'\in \wset_j}|\out(R')|\leq 2.4rk\log r$. Therefore, at least one subset in $\wset^1$  contains more than $\frac{2^{11}r\log r\cdot F(k/2)}{5r\log r}>2^7\cdot F(k/2)$ vertices.
 \end{proof}

Next, we try to route the terminals in $\tset$ to the edges in $\out(S_j')$, as follows. We build a flow network, starting from the graph $G'$, contracting all terminals in $\tset$ into a source $s$, and all vertices in $S_j'$ 
into a sink $t$. We try to find an $s$-$t$ flow in this network of value at least $\ceil{k/2}$. Assume first that we are unable to find such flow. Then we can find a minimum $s$-$t$ cut $(A,B)$, with $s\in A$, $t\in B$, and the cut value is less than $\ceil{k/2}$ in this network. Let $B'=(B\setminus\set{t})\cup S'_j$. Then $B'\sse V(G')\setminus \tset'$, and it is a contractible set, since $|\out_{G'}(B')|\leq \ceil{k/2}$, while $S'_j\sse B'$, so $|B'|>2^7\cdot F(k/2)$.
Moreover, $G'[B']$ is a connected graph, since $G'[S_j]$ is connected.

Assume now  that we have managed to find a flow of value at least $\ceil{k/2}$ in this network. Then this defines a collection $\pset_j'$ of $\ceil{k/2}$ edge-disjoint paths, connecting distinct terminals in $\tset$ to distinct edges of $\out(S_j')$.

Overall, our algorithm may terminate early, in which case it is guaranteed to produce either a type-2 witness $\tA$, or a contractible set. Otherwise, the algorithm finds a family $\fset=\set{S_1',\ldots,S_r'}$ of vertex subsets that are $\alpha_W(k^*)$-well-linked, with the collections $\pset_1',\ldots,\pset_r'$ of paths as required, thus giving a type-1 witness. 
\end{proof}

We are now ready to complete the proof of Theorem~\ref{thm: flow sparsifier: well-linked terminals}. 
If the set $R=V(G)\setminus \tset$ is a good router, then we let $\cset=\set{R}$, and the sparsifier $H$ is the corresponding contracted graph (a star graph, where the star center is $v_R$, and the leaves are the terminals). Assume now that $R$ is not a good router. We then start with $G'=G$, and repeatedly apply Theorem~\ref{thm: can find good structures or contract further} to $G'$.
From Theorem~\ref{thm: from witnesses to good routers}, graph $G'$ cannot contain type-1 or type-2 witnesses, so the output of the theorem will always be a contractible set $S$ of vertices in $G'$. We then apply Procedure $\contract(G',S)$ to obtain a new contracted graph $G''$, with $|V(G'')|<|V(G')|$, and continue. We are guaranteed to obtain, after at most $|V(G)|$ iterations, a legal contracted graph $G'$ with $|V(G')\setminus\tset|\leq F(k)$, which we output as the final sparsifier. From Claim~\ref{claim: a legal contracted graph is a sparsifier}, $G'$ is indeed a quality-$(2\eta^*)$-sparsifier.

In order to bound the running time of the algorithm, we prove that for any $n$-vertex graph with a set $\tset$ of $k$ terminals, the running time of the algorithm is $T(n,k)=n^{O(\log k)}\cdot 2^k$. The proof is by induction on the values of $k$. For $k\leq 4$, the running time of the algorithm is $\poly(n)$. Assume that the claim holds for all values $k'< k$, and we now prove it for $k$.
Recall that our algorithm performs at most $n$ iterations. Each iteration involves a call to procedure $\contract(G',S)$, and takes an additional time of $\poly(n)$. Let $k'=|\out(S)|$, and recall that $k'\leq \ceil{k/2}$. Procedure $\contract(G',S)$ computes a strong well-linked decomposition $\sset$ of the set $S$, and then computes a sparsifier for each set $Z\in \sset$ recursively. For each set $Z\in \sset$, let $k_Z=\out_G(Z)$, and let $n_Z=|Z|$. Then $k_Z\leq k'\leq \ceil{k/2}$, and $\sum_{Z\in \sset}n_Z\leq n$. Therefore, by the induction hypothesis, the running time of the recursive procedure for each set $Z\in \sset$ is at most $T(n_Z,k_Z)\leq T(n_Z,\ceil{k/2})$, and the total running time of procedure $\contract(G',S)$ is at most $2^k\poly(n)+\sum_{Z\in \sset}T(n_Z,\ceil{k/2})\leq 2^k\poly(n)+T(n,\ceil{k/2})$. Overall, the running time of the algorithm is then bounded by $n\cdot (\poly(n)+2^k\poly(n)+T(n,\ceil{k/2}))\leq 2^k\poly(n)+n\cdot (n^{O(\log \ceil{k/2}})\cdot 2^k)\leq n^{O(\log k)}\cdot 2^k=T(n,k)$.
%\end{proof}

%The algorithm has at most $|V(G)|$ iterations, and the running time of each iteration is dominated by procedure $\contract(G',S)$, whose running time is $\poly(n)\cdot 2^k$, since it computes a strong well-linked decomposition. Overall, the running time of the algorithm is bounded by $\poly(n)\cdot 2^k$.
% \end{proof}

\paragraph{Acknowledgements} 

The author thanks Yury Makarychev and Konstantin Makarychev for many interesting discussions about vertex sparsifiers.

%-------------------------------------------------------------------------------------------------------
%-------------------------------------------------------------------------------------------
%-------------------------------------------------------------------------------------------
%--------------------------------------------------------------------------------------------------
%--------------------------------------------------------------------------------------------------
%--------------------------------------------------------------------------------------------------
%--------------------------------------------------------------------------------------------------

\bibliography{vertex-sparsifiers}
\bibliographystyle{alpha}

%\newpage

%--------------------------------------------------------------------------------------------------
%--------------------------------------------------------------------------------------------------
%--------------------------------------------------------------------------------------------------
\appendix
%--------------------------------------------------------------------------------------------------
%--------------------------------------------------------------------------------------------------
%--------------------------------------------------------------------------------------------------

\label{------------------------------------------appendix----------------------------------------}
%-------------------------------------------------------------------------------------------
%-------------------------------------------------------------------------------------------
%-------------------------------------------------------------------------------------------
%\newpage
\iffull
\section{List of Parameters}
\label{sec: parameters list}
\begin{tabular}{|l|l|l|}\hline
&&\\
$\alphasc(k)$&$O(\sqrt{\log k})$&Approximation factor of algorithm $\algsc$ for sparsest cut\\ \hline
&&\\
$\alpha_W(z)$&$\Omega\left(\frac 1 {\log^{3/2} z}\right)$&Well-linkedness parameter for the\\
&&weak well-linked decomposition\\ \hline
&&\\
$\betaFCG(k) $&$O(\log k)$&Flow-cut gap for undirected graphs \\ \hline
&&\\
$\eta^*$&$34$&Parameter from the definition of good routers\\ \hline
&&\\
$r$&$O(\log^3k)$&Number of vertex subsets in a type-1 witness\\ \hline
&&\\
$k^*$& $2kr\log r=k\poly\log k$& \\ \hline
\end{tabular}
\fi

\label{------------------------------------------proofs from prelims---------------------------}
%--------------------------------------------------------------------------------------------------
%--------------------------------------------------------------------------------------------------
%--------------------------------------------------------------------------------------------------
\section{Proofs Omitted from Section~\ref{sec: Prelims}}
%---------------------------------------------------------------------
%---------------------------------------------------------------------
%---------------------------------------------------------------------
%---------------------------------------------------------------------
\subsection{Proof of Theorem~\ref{thm: weak well-linked}}
\label{subsec: proof for weak well-linked}
%---------------------------------------------------------------------
%---------------------------------------------------------------------
%---------------------------------------------------------------------
%---------------------------------------------------------------------

%---------------------------------------------------------------------
%---------------------------------------------------------------------
%\paragraph{Proof of Theorem~\ref{thm: weak well-linked}}
%---------------------------------------------------------------------
%---------------------------------------------------------------------

We use the $\alphasc(z)$-approximation algorithm $\algsc$ for the sparsest cut problem. We set $\alpha_W(z)=\frac 1{2^{7}\alphasc(z)\log z}=\Omega\left(\frac 1 {\log^{3/2} z}\right)$.

Throughout the algorithm, we maintain a partition $\wset$ of the input set $S$ of vertices, where for each $R\in \wset$, $|\out(R)|\leq |\out(S)|$. At the beginning, $\wset$ consists of the subsets of $S$ defined by the connected components of $G[S]$. 

Let $R\in \wset$ be any set in the current partition, and let $(G_R,\tset'_R)$ be the instance of the sparsest cut problem corresponding to $R$, as defined in Section~\ref{sec: Prelims}.
We say that a cut $(A',B')$ in $G_R$ is sparse, iff its sparsity is less than $1/(2^{7}\log z)$. We apply the algorithm $\algsc$ to the instance $(G_R,\tset'_R)$ of sparsest cut. If the algorithm returns a cut $(A',B')$, that is a sparse cut, then let $A=A'\setminus \tset'_R$, and $B=B'\setminus \tset'_R$. We remove $R$ from $\wset$, and add $A$ and $B$ to it instead. Let $T_A=\out(R)\cap \out(A)$, and $T_B=\out(R)\cap \out(B)$, and assume w.l.o.g. that $|T_A|\leq |T_B|$. Then $|E(A,B)|< |T_A|/(2^{7}\log z)$ must hold, and in particular, $|\out(A)|\leq |\out(B)|\leq |\out(R)|\leq |\out(S)|$. For accounting purposes, each edge in set $T_A$ is charged $1/(2^{7}\log z)$ for the edges in $E(A,B)$. Notice that the total charge to the edges in $T_A$ is $|T_A|/(2^{7}\log z)\geq |E(A,B)|$. Notice also that since $|T_A|\leq |\out(R)|/2$ and $|E(A,B)|\leq |T_A|/2^{7}$, $|\out(A)|\leq 0.51|\out(R)|$.

The algorithm stops when for each set $R\in \wset$, the procedure $\algsc$ returns a cut that is not sparse. We argue that this means that each set $R\in \wset$ is $\alpha_W(z)$-well-linked. Assume otherwise, and let $R\in \wset$ be a set that is not $\alpha_W(z)$-well-linked. Then, by the definition of well-linkedness, the corresponding instance of the sparsest cut problem must have a cut of sparsity less than $\alpha_W(z)=1/(2^{7}\alphasc(z)\log z)$. The algorithm $\algsc$ should then have returned a cut whose sparsity is less than $\alpha_W(z)\cdot \alphasc(z)=1/(2^{7}\log z)$, that is a sparse cut.

Finally, we need to bound $\sum_{R\in\wset}|\out(R)|$. We use the charging scheme defined above. Consider some iteration where we partition the set $R$ into two subsets $A$ and $B$, with $|T_A|\leq |T_B|$. Recall that each edge in $T_A$ is charged $1/(2^{7}\log z)$ in this iteration, while $|\out(A)|\leq 0.51|\out(R)|$ holds.
Consider some edge $e=(u,v)$. Whenever $e$ is charged via the vertex $u$, the size of the set $\out(R)$, where $u\in R\in \wset$ goes down by a factor of at least $0.51$. Therefore, $e$ can be charged at most $2\log z$ times via each of its endpoints. The total charge to $e$ is then at most $4\log z/(2^{7}\log z)=1/2^5$. This however only accounts for the \emph{direct} charge. For example, some edge $e'\not \in \out(S)$, that was first charged to the edges in $\out(S)$, can in turn be charged for some other edges. We call such charging \emph{indirect}. If we sum up the indirect charge for every edge $e\in \out(S)$, we obtain a geometric series, and so the total direct and indirect amount charged to every edge $e\in \out(S)$ is at most $1/2^4<0.1$. Therefore, $\sum_{R\in \wset}|\out(R)|\leq 1.2|\out(S)|$ (we need to count each edge $e\in \left (\bigcup_{R\in \wset}\out(R)\right )\setminus \out(S)$ twice: once for each its endpoint).\hfill 

%---------------------------------------------------------------------
%---------------------------------------------------------------------
%---------------------------------------------------------------------
%---------------------------------------------------------------------
\subsection{Proof of Theorem~\ref{thm: strong well linked}}
\label{subsec: proof for strong well-linked}
%---------------------------------------------------------------------
%---------------------------------------------------------------------
%---------------------------------------------------------------------
%---------------------------------------------------------------------
%\paragraph{Proof of Theorem~\ref{thm: strong well linked}}
%---------------------------------------------------------------------
%---------------------------------------------------------------------

The algorithm  is very similar to the algorithm used in the proof of Theorem~\ref{thm: weak well-linked}, but the analysis is different. We start by describing the algorithm.

Throughout the algorithm, we maintain a partition $\sset$ of the input set $S$ of vertices, where for each $R\in \sset$, $|\out(R)|\leq |\out(S)|$. At the beginning, $\sset=\set{S}$. 

Let $R\in \sset$ be any set in the current partition, and let $(G_R,\tset'_R)$ be the corresponding instance of the sparsest cut problem.
We say that a cut $(A',B')$ in $G_R$ is sparse, iff its sparsity is less than $1/3$. Notice that the set $R$ is $1/3$-well-linked iff there is no sparse cut in $R$

Our algorithm proceeds as follows.
 Whenever there is a set $R\in \sset$, such that the sparsity of the sparsest cut $(A',B')$ in the corresponding instance $(G_R,\tset'_R)$ is less than $1/3$, we set $A=A'\setminus \tset'_R$ and $B=B'\setminus \tset'_R$. We then remove $R$ from $\sset$, and add $A$ and $B$ to it instead.
Since $|\out(R)|\leq |\out(S)|\leq z$, the sparsest cut problem instance can be solved in time $2^z\poly(n)$: we simply go over all bi-partitions $(\tset_1,\tset_2)$ of the set $\tset'_R$ of terminals, and for each such bi-partition, compute the minimum cut separating the vertices in $\tset_1$ from the vertices in $\tset_2$. 
The algorithm terminates when for every set $R\in \sset$, the value of the sparsest cut in the corresponding instance is at least $1/3$.
From the above discussion, the running time of the algorithm is $2^z\poly(n)$. It is also clear that when the algorithm terminates, for every set $R\in \sset$, $|\out(R)|\leq |\out(S)|$, and $R$ is $1/3$-well-linked.
It now only remains to analyze the sizes of the collections $\sset_i$ of vertex subsets in the resulting partition, and to bound $\sum_{R\in \sset}|\out(R)|$.

 Let $\cset$ be the set of all cuts that we produce throughout this algorithm (that is, $\cset$ contains all sets $R$ that belonged to $\sset$ at any stage of the algorithm), and let $T$ be the corresponding partitioning tree, whose vertices represent the sets in $\cset$, and every inner vertex $v_R$ has exactly two children $v_A,v_B$, where $(A,B)$ is the partition that our procedure computed for the set $R$.

For the sake of the analysis of the algorithm, we define a new graph $G'$, as follows. We start with the graph $G$, and  the final partition $\sset$ of $S$.
Let $R,R'\in \sset$ be any pair of distinct vertex subsets, and let $e=(u,v)$ be any edge with $u\in R,v\in R'$. We subdivide the edge $e$ by adding two vertices $u_e,v_e$ to it, so that now we have a path $(u,u_e,v_e,v)$ instead of the edge $e$. Additionally, for every edge $e=(u,v)$ with $u\in S$, $v\not\in S$, we subdivide edge $e$ by adding a new vertex $u_e$ to it.

All the newly added vertices are called \emph{terminals} and the set of all such terminals is denoted by $\Gamma$. The resulting graph is denoted by $G'$.
For each subset $X\sse S$ of vertices, we will still denote by $\out(X)=\out_G(X)$. Consider now some subset $R\in \cset$ of vertices. We now define the subset $\Gamma(R)\sse \Gamma$ of terminals associated with $R$, as follows: $\Gamma(R)=\set{u_e\mid e=(u,v)\in \out(R), u\in R}$, so $|\Gamma(R)|=|\out(R)|$. Moreover, $\sum_{R\in \sset}|\out(R)|=\sum_{R\in \sset}|\Gamma(R)|=|\Gamma|$.
 
We define a charging scheme that will help us bound the number of terminals, and the number of sets in each collection $\sset_i$. 
 
%Therefore, throughout the algorithm, each set $R\in \sset$ is associated with a set $\Gamma(R)$ of vertices, which contains one vertex for each edge in $\out_G(R)$. At every step, we let $G'$ denote the resulting graph (after the subdivision of the edges), and $G'[R]$ the sub-graph of $G'$ induced by $R\cup \Gamma(R)$. 
%We let $\Gamma'(S)=\Gamma(S)\cup (S\cap T)$, and we refer to $\Gamma'(S)$ as the \emph{terminals} of $S$.

%For each set $R\in \sset$, we can set up an instance of the non-uniform sparsest cut problem on the graph $G'[R]$, with the vertices in $\Gamma(R)$ serving as terminals, whose weights are $1$.
%, we use $\out(R)$ to denote $\out_G(R)$, for any subset $R\in \sset$ of vertices. Notice that $|\out(R)|=|\Gamma(R)|$ always holds.

We say that a set $R\in \cset$ belongs to level $i$ iff $\frac z{2^{i/2}}< |\out(R)|\leq \frac z {2^{(i-1)/2}}$. In particular, $S$ is a level-$1$ set, and we have at most $2\log z+1$ levels.
We also partition all terminals into levels, and within each level, we have two types of terminals: regular and special.
Intuitively, special terminals at level $i$ are the terminals that have been created by partitioning some sets from levels $1,\ldots,i-1$, while regular terminals are created by partitioning sets that belong to level $i$.

The terminals in set $\Gamma(S)$ are called special terminals, and they belong to level $1$.
Let $R$ be some level-$i$ set, and let $A,B$ be its children, with $|\out(R)\cap \out(A)|\leq |\out(R)\cap \out(B)|$. Then $|E(A,B)|\leq |\out(R)\cap \out(A)|/3=|\Gamma(R)\cap \Gamma(A)|/3$, and so $|\Gamma(A)|= |\Gamma(R)\cap \Gamma(A)|+|E(A,B)|\leq |\Gamma(R)|(\half+\frac 1 6)< |\Gamma(R)|/\sqrt 2$. Let $i'$ and $i''$ be the levels to which sets $A$ and $B$ belong, respectively. Then $i'\geq i+1$ must hold. We call all terminals in set $\Gamma(A)\setminus \Gamma(R)$ \emph{special terminals for level $i'$} (to indicate that they were created from partitioning a lower-level set). If $i''\geq i+1$ holds as well, then we call all terminals in set  $\Gamma(B)\setminus \Gamma(R)$ special terminals at level $i''$. Otherwise, they are \emph{regular terminals}, that belong to level $i''=i$.

Notice that both sets $\Gamma(A)\setminus \Gamma(R)$  and $\Gamma(B)\setminus \Gamma(R)$ of new terminals come from subdivisions of the edges in $E(A,B)$, so each such edge gives rise to one terminal in $\Gamma(A)$ and one in $\Gamma(B)$. We stress that for each $i$, a level-$i$ terminal continues to be a level-$i$ terminal throughout the algorithm, even if its corresponding subset of vertices in $\sset$ becomes further subdivided and stops being a level-$i$ set. So for example, if $R$ is a level-$i$ set, the terminals in $\Gamma(R)$ may belong to levels $1,\ldots,i$.
The next lemma bounds the number of terminals at each level, and it is central to the analysis of the algorithm.

\begin{lemma}
For each $1< i\leq 2\log z+1$, there are at most $2^{i-2}z$ special terminals, and at most $2^{i-1} z$ regular level-$i$ terminals. For $i=1$, there are $z$ level-$1$ special terminals, and at most $z/2$ regular level-$1$ terminals. 
\end{lemma}

\begin{proof}
For each level $i$, let $n_i$ be the number of special terminals, and $n'_i$ the number of regular terminals. Let $S_i$ be the total number of terminals at levels $1,\ldots, i$.
We use the following two simple claims.
\begin{claim}
For each $i>1$, $n_i\leq S_{i-1}/3$.
\end{claim}
\begin{proof}
Recall that a level-$i$ special terminal can only be created when partitioning some cluster $R\in \cset$ that belongs to levels $1,\ldots,i-1$. Suppose that the partition of $R$ is $(A,B)$, and assume that $|\Gamma(A)\cap \Gamma(R)|\leq |\Gamma(B)\cap \Gamma(R)|$. Let $X\in \set{A,B}$ be the cluster that belongs to level-$i$ (it is possible that both $A$ and $B$ belong to level $i$ - the analysis for this case is similar and is carried out for each one of the clusters separately). Then the terminals in $\Gamma(X)\cap \Gamma(R)$ belong to levels $1,\ldots,i-1$, and the terminals in $\Gamma(X)\setminus \Gamma(R)$ become special terminals at level $i$. We charge the terminals in $\Gamma(X)\cap \Gamma(R)$ for the terminals in $\Gamma(X)\setminus \Gamma(R)$, where the charge to every terminal in $\Gamma(X)\cap \Gamma(R)$ is at most $1/3$. Moreover, since $X$ now becomes a level-$i$ cluster, the terminals in $\Gamma(X)$ will never be charged for special level-$i$ terminals again. Therefore, the number of special level-$i$ terminals is at most $S_{i-1}/3$.
 \end{proof}

\begin{claim}\label{claim: bound on regular}
For every level $i$, $n'_i\leq (S_{i-1}+n_i)/2$.
\end{claim}
\begin{proof}
The proof uses a charging scheme, that is defined as follows. 
Let $\Gamma'$ be the set of all terminals at levels $1,\ldots,i-1$, together with the special terminals at level $i$. We charge the regular level-$i$ terminals to the terminals in $\Gamma'$, and show that $n'_i\leq |\Gamma'|/2$.

Recall that regular level-$i$ terminals are only created by partitioning level-$i$ clusters $R\in \cset$. Let $R$ be any such level-$i$ cluster, that we have partitioned into $(A,B)$. Assume w.l.o.g. that $|T_A|\leq |T_B|$, and recall that $A$ must belong to some level $i'>i$. This partition only creates level-$i$ terminals if $B$ belongs to level $i$. The number of the newly created level-$i$ terminals, $|\Gamma(B)\setminus \Gamma(R)|=|E(A,B)|\leq |T_A|/3=|\Gamma(A)\cap \Gamma(R)|/3$. We charge the terminals in $|\Gamma(A)\cap \Gamma(R)|$ for the newly created terminals in $\Gamma(B)\setminus \Gamma(R)$, where each terminal in $|\Gamma(A)\cap \Gamma(R)|$ is charged $1/3$. Observe that the terminals in $|\Gamma(A)\cap \Gamma(R)|$ must belong to levels $1,\ldots, i$, and moreover, since $A$ does not belong to level $i$, these terminals will never be charged again for level-$i$ terminals. 
But we may charge them indirectly - by charging the terminals in $\Gamma(B)\setminus \Gamma(R)$ for some new terminals. However, since a direct charge to every terminal is bounded by $1/3$, the total direct and indirect charge to any terminal in $\Gamma'$ forms a geometrically decreasing sequence, and its sum is bounded by $1/2$. Therefore, $n'_i\leq |\Gamma'|/2$.
 \end{proof}

The number of special level-$1$ terminals is $z$, and  the number of regular level-$1$ terminals is at most $z/2$ by Claim~\ref{claim: bound on regular}, so $S_1\leq 3z/2$. In general, we now have that for any $i>1$, $n_i\leq S_{i-1}/3$, and $n_{i}'\leq (S_{i-1}+n_i)/2\leq 2S_{i-1}/3$. Therefore, $S_i=S_{i-1}+n_i+n_i'\leq 2S_{i-1}$. We conclude that $S_i\leq  3\cdot 2^{i-2}z$ for all $i$, and so $n_i\leq 2^{i-2}z$, and $n_i'\leq 2^{i-1}z$  for all $i> 1$. 
 \end{proof}

Consider now the final partition $\sset$. We can now bound $\sum_{R\in \sset}|\out(R)|=|\Gamma|=|S_{2\log z+1}|\leq 3z\cdot 2^{2\log z}=O(z^3)$.

Let $R\in \sset$, and assume that $R$ belongs to level $i$. Then all terminals of $\Gamma(R)$ must belong to levels $1,\ldots, i$. The total number of such terminals is bounded by $3z\cdot 2^{i-2}$, and set $R$ uses at least $z/2^{i/2}$ of them. It follows that the number of level-$i$ sets in $\sset$ is bounded by $2^{3i/2}$.

Recall that $\sset_i$ contains all sets $R\in \sset$ with  $z/2^{i}< |\out(R)|\leq z/2^{i-1}$, and so $\sset_i$ only contains sets $R\in \sset$ that belong to levels $2i$ or $2i-1$. From the above discussion $|\sset_i|\leq 2^{3i+3}$.

%\paragraph{Remarks:}
%This is not optimized. We can balance the $\alpha$ factor and the power of $k$ we obtain in $|V'|$. We can also get $|V'|=d\poly(k)$ and a constant quality of sparsification: in the original graph $G$, replace every vertex by a constant-degree expander. The vertices of the expanders that replaced the terminals become the terminals in the new graph. Finally, we observe that the above proof is non-constructive, since we need to solve the sparsest cut problem.\hfill 

\label{-------------------------------------------Proofs for Flow Sparsifiers-----------------------------------}
%------------------------------------------------------------------------------
%---------------------------------------------------------------------------------
%\section{Proofs Omitted from Section~\ref{sec: flow sparsifiers}}\label{sec: proofs for flow sparsifiers}
%------------------------------------------------------------------------------
%---------------------------------------------------------------------------------
%------------------------------------------------------------------------------
%---------------------------------------------------------------------------------

%-------------------------------------------------------------------------------------
\section{Proof of Corollary~\ref{corollary: general flow sparsifier}} \label{subsec: Proof of the corollary}
%-------------------------------------------------------------------------------------
%-------------------------------------------------------------------------------------

Assume first that all edge capacities are integral and bounded by $C$. For each terminal $t\in \tset$, let $C_t$ be the total capacity of all edges incident on $t$. We replace every edge $e\in E$ with $c_e$ parallel edges of unit capacity. For each terminal $t\in \tset$, we sub-divide each edge $e$ incident on $t$ with a vertex $v_e$, and we let $S_t$ be the set of these new vertices. Let $\tilde{G}$ be the resulting graph, and let $G'=\tilde{G}\setminus \tset$. We let $\tset'=\bigcup_{t\in \tset}S_t$ be the set of terminals for the new graph $G'$. Then $|\tset'|=C$, and each vertex in $\tset'$ has exactly one edge incident to it. We now apply Theorem~\ref{thm: flow sparsifier for unit capacities} to $(G',\tset')$ to obtain a sparsifier $H'$. In our final step, for each $t\in \tset$, we unify all vertices in the set $S_t$ in graph $H'$ into a single vertex $t$. Let $H$ be this final sparsifier. 
Clearly, $|V(H)|=C^{O(\log\log C)}$. We now show that $H$ is a quality-$(2\eta^*)$ sparsifier for $G$. 

Let $D$ be any set of demands on $\tset$, and let $F$ be the routing of these demands in graph $G$ with congestion $\eta=\eta(G,D)$. Then the routing $F$ naturally defines a set $D'$ of demands over the vertices of $\tset'$. For each pair $(v_e,v_{e'})\in \tset'$ of vertices, the demand $D'(v_e,v_{e'})$ is the total flow on all flow-paths in $F$ that start at $e$ and terminate at $e'$. Flow $F$ also gives a routing of this new set $D'$ of demands in graph $G'$ with congestion $\eta$. Since $H'$ is a flow sparsifier for $G'$, there is a routing $F'$ of the demands in $D'$ in graph $H'$ with congestion at most $\eta$. This routing induces a routing of the set $D$ of demands in graph $H$ with congestion at most $\eta$.

The other direction is proved similarly: if $D$ is any set of demands in $\tset$, and $F$ is the routing of these demands in $H$ with congestion $\eta=\eta(H,D)$, then $F$ defines a set $D'$ of demands on the vertices of $\tset'$ exactly as before. Flow $F$ then induces a routing of the set $D'$ of demands in graph $H'$ with congestion at most $\eta$. Since $H'$ is a quality-$(2\eta^*)$ sparsifier for $G'$, there is a routing $F'$ of the set $D'$ of demands in graph $G'$ with congestion at most $\eta\cdot 2\eta^*$. Flow $F'$ then induces a routing of the set $D$ of demands in graph $G$ with congestion at most $\eta\cdot 2\eta^*$.

Finally, consider the general case, where the edge capacities $c_e\geq 1$ are no longer required to be integral, and may not be bounded by $C$. For each edge $e\in E$, if $c_e>C$, then we set $c_e=C$. It is easy to see that this transformation does not affect the values $\eta(G,D)$ for demand sets $D$, since all flow in the network must traverse one of the edges incident to the terminals. Next, we define new edge capacities, by setting $c'_e=\ceil{ \frac{2\eta^*}{\eps}c_e}$. Let $G'$ be the resulting graph. Notice that for each edge $e$, $\frac{2\eta^*}{\eps}c_e\leq c'_e\leq \frac{2\eta^*}{\eps}c_e+1\leq \frac{2\eta^*+\eps}{\eps}c_e$. 

Consider some set $D$ of demands defined over the set $\tset$ of terminals. Let $F$ be the routing of $D$ in graph $G'$, whose congestion is $\eta(G',D)$. Consider the same flow $F$ in graph $G$. The congestion caused by $F$  on each edge $e$ of $G$ is bounded by $\frac{F(e)}{c_e}\leq \frac{F(e)}{c'_e}\cdot \frac{2\eta^*+\eps}{\eps}$. Therefore, $\eta(G,D)\leq \eta(G',D)\cdot \frac{2\eta^*+\eps}{\eps}$ for all $D$.

Similarly, given any set $D$ of demands, let $F$ be the routing of $D$ in graph $G$, whose congestion is $\eta(G,D)$. Consider the same flow $F$ in graph $G'$. The congestion on each edge $e$ in $G'$ is bounded by $\frac{F(e)}{c'_e}\leq \frac{F(e)}{c_e}\cdot \frac{\eps}{2\eta^*}$. Therefore, $\eta(G',D)\leq \eta(G,D)\cdot \frac{\eps}{2\eta^*}$ for all $D$. We conclude that for all $D$, $\frac{\eps}{2\eta^*+\eps}\eta(G,D)\leq \eta(G',D)\leq \frac{\eps}{2\eta^*}\eta(G,D)$.

%Therefore, for each set $D$ of demands, $\frac{2\eta^*}{\eps}\eta(G,D)\leq \eta(G',D)\leq \frac{2\eta^*+\eps}{\eps}\eta(G,D)$.

Let $H'$ be a quality-$(2\eta^*)$ sparsifier for graph $G'$, and let $H$ be the graph obtained from $H'$ by multiplying all edge capacities by the factor of $\frac{\eps}{2\eta^*}$. Then for any set $D$ of demands:

\[\eta(H,D)=\frac{2\eta^*}{\eps}\eta(H',D)\leq\frac{2\eta^*}{\eps}\eta(G',D)\leq \eta(G,D);\]

and

\[\begin{split}
\eta(H,D)&=\frac{2\eta^*}{\eps}\eta(H',D)\\
&\geq \frac 1{2\eta^*}\frac{2\eta^*}{\eps}\eta(G',D)
\\ &\geq \frac{1}{2\eta^*+\eps}\eta(G,D).\end{split}\]

We conclude that $H$ is a quality-$(2\eta^*+\eps)$-sparsifier for $G$, and $|V(H)|=C^{O(\log\log C)}$.

\section{Proof Omitted from Section~\ref{subsec: Proof for well-linked terminals}}
%------------------------------------------------------------------------------
%---------------------------------------------------------------------------------
\subsection{Proof of Claim~\ref{claim: well-linkedness of sets in G}}
%------------------------------------------------------------------------------
%---------------------------------------------------------------------------------
%------------------------------------------------------------------------------
%---------------------------------------------------------------------------------
For simplicity, we build two new graphs $H$ and $H'$, as follows. Sub-divide every edge $e\in E'$ with a vertex $v_e$ in both $G$ and $G'$, and let $\tset'=\set{v_e\mid e\in E'}$. Let $H$ be the sub-graph of the resulting graph obtained from $G$, induced by $S\cup \tset'$, and let $H'$ be the sub-graph of the corresponding graph obtained from $G'$, induced by $S'\cup \tset'$. Clearly, $H'$ is a legal contracted graph for $H$, and from the definition of well-linkedness, $H'$ is $\alpha$-well-linked for $\tset'$. We only need to prove that $H$ is $\alpha/3$-well-linked for $\tset'$. For any subset $Z$ of vertices in either $H$ or $H'$, let $\tset_Z=Z\cap \tset'$.

Assume for contradiction that $H'$ is not $\alpha$-well-linked for $\tset'$, and let $(X,Y)$ be the violating partition of $V(H)$, that is, $|E_H(X,Y)|<\frac{\alpha}{3}\min\set{|\tset_X|,|\tset_Y|}$. We construct a partition $(X',Y')$ of $V(H')$, such that $|E_{H'}(X',Y')|<\alpha\cdot\min\set{|\tset_{X'}|,|\tset_{Y'}|}$, contradicting the fact that $H'$ is $\alpha$-well-linked for $\tset'$.

Let $\cset_S\sse \cset$ be the collection of all clusters $C$, with $v_C\in S'$. In order to construct the partition $(X',Y')$ of $V(H')$, we start with the partition $(X,Y)$ of $V(H)$, and process all clusters $C\in \cset_S$ one-by-one. For each such cluster, we  move all vertices of $C$ either to $X$ or to $Y$. Once we process all clusters in $\cset_S$, we will obtain obtain a partition $(\tilde X,\tilde Y)$ of $V(H)$, that will naturally define a partition $(X',Y')$ of $V(H')$.

Consider some cluster $C\in \cset_S$, and partition the edges in $\out_{H}(C)$ into four subsets, $E_X,E_Y$,  $E_{XY},E_{YX}$, as follows. For each edge $e=(u,v)\in \out_H(C)$, with $u\in C$, $v\not\in C$, if both $u,v\in X$, then we add $e$ to $E_X$; if both $u,v\in Y$, we add $e$ to $E_Y$; if $u\in X$, $v\in Y$, then we add $e$ to $E_{XY}$, and otherwise we add it to $E_{YX}$.  If $|E_X|+|E_{XY}|\leq |E_Y|+|E_{YX}|$, then we move all vertices of $C$ to $Y$, and otherwise we move them to $X$. Let $E_C\sse E_H(X,Y)$ be the subset of the edges in the cut $(X,Y)$, with both endpoints in $C$, that is, $E_C=E_H(X\cap C, Y\cap C)$.

Assume w.l.o.g. that $|E_X|+|E_{XY}|\leq |E_Y|+|E_{YX}|$, and so we have moved the vertices of $C$ to $Y$. The only new edges that have been added to the cut are the edges of $E_X$. Since set $C$ is $1/3$-well-linked in $H$, $|E_X|\leq 3|E_C|$. We charge the edges in $E_C$ for the edges in $E_X$. The charge to every edge of $E_C$ is at most $3$, and since the edges of $E_C$ have both endpoints inside the cluster $C$, we will never charge them again. 

Once we process all super-nodes $v_C$ in this fashion, we obtain a partition $(\tilde{X},\tilde{Y})$ of $V(H)$, where for every cluster $C\in \cset_S$, all vertices of $C$ belong to either $\tilde X$ or $\tilde Y$. This cut naturally defines a partition $(X',Y')$ of the vertices of $V(H')$, where $v_C\in X'$ iff $C\sse \tilde X$. From the above discussion, $|E_{H'}(X',Y')|=|E_{H}(\tilde X,\tilde Y)|\leq 3|E_{H}(X,Y)|$. Moreover, since the terminals in $\tset'$ do not belong to any cluster $C\in \cset_S$, the partitions of the terminals of $\tset'$ induced by the cuts $(X,Y)$ in $H$ and $(X',Y')$ in $H'$ are identical. Therefore, $|E_{H'}(X',Y')|<\alpha\min\set{|\tset_{X'}|,|\tset_{Y'}|}$, a contradiction.

%-------------------------------------------------------------------------------
%-------------------------------------------------------------------------------
%-------------------------------------------------------------------------------
\subsection{Proof of Claim~\ref{claim: paths in contracted and original graphs}}

We set up the following two flow networks. For the first flow network, we start with the graph $G$. We add a source $s$, and connect it with a directed edge to every vertex $v\in \tset'$. For every edge $e\in E'$, we sub-divide $e$ by adding a vertex $z_e$ to it, and connect $z_e$ to the sink $t$ with a directed edge. We set the capacity of every edge in this network to be $3\eta$, except for the edges leaving $s$ or entering $t$, whose capacities are set to $1$. Let $N_1$ denote the resulting flow network. In order to show the existence of the set $\pset$ of paths, it is enough to show that the value of the maximum flow in network $N_1$ is $|\tset'|$. For each edge $e\in E(N_1)$, we denote by $c(e)$ its capacity, and for each cut $(X,Y)$ in the network, we denote by $c(X,Y)$ the total capacity of edges connecting the vertices of $X$ to the vertices of $Y$.

The second network, $N_2$, is constructed similarly, except that we use the contracted graph $G'$, instead of the graph $G$. Specifically, we start with graph $G'$, add a source $s$ and a sink $t$. Source $s$ connects with a directed edge to every vertex $v\in\tset'$. As before, we sub-divide every edge $e\in E'$ with a vertex $z_e$, and connect $z_e$ to the sink $t$. The capacities of all edges in $N_2$ are $\eta$, except for the edges that leave the source $s$ or enter the sink $t$, whose capacities are set to $1$. Observe that the existence of the set $\pset'$ of paths in graph $G'$ guarantees that there is a flow of value $|\tset'|$ in network $N_2$.  For each edge $e\in E(N_2)$, we denote by $c'(e)$ its capacity in $N_2$, and for each cut $(X,Y)$ in the network, we denote by $c'(X,Y)$ the total capacity of edges connecting the vertices of $X$ to the vertices of $Y$.

It is now enough to prove that there is a flow of value $|\tset'|$ in network $N_1$. Assume this is not the case. Then there is an $s$-$t$ cut $(X,Y)$ in network $N_1$, with $c(X,Y)<|\tset'|$. We show that there is an $s$-$t$ cut $(X',Y')$ in network $N_2$, with $c'(X',Y')<|\tset'|$, contradicting the existence of the set $\pset$ of paths.

We consider the super-nodes $v_C\in V(G')$ one-by-one. For each such super-node $v_C$, we move all vertices of $C$ either to $X$ or to $Y$. The final cut, $(\tilde X,\tilde Y)$ will naturally define an $s$-$t$ cut $(X',Y')$ in network $N_2$, and we will show that its capacity is less than $|\tset'|$.

Consider some super-node $v_C\in V(G')$. Recall that we are guaranteed that $\tset'\cap C=\emptyset$, so $\out_{N_1}(C)=\out_{G}(C)$. Let $E_C=E_{G}(C\cap X,C\cap Y)$.
We partition the edges of $\out_{G}(C)$ into four subsets, $E_X$ containing edges with both endpoints in $X$, $E_Y$ containing edges with both endpoints in $Y$, $E_{XY}$ containing edges $e=(u,v)$ with $u\in X\cap C$, $v\in Y\setminus C$, and $E_{YX}$ containing the remaining edges. If $|E_X|+|E_{XY}|\leq |E_Y|+|E_{YX}|$, then we move all vertices of $C$ to $Y$, and otherwise we move them to $X$.

Assume w.l.o.g. that $|E_X|+|E_{XY}|\leq |E_Y|+|E_{YX}|$, so we have moved the vertices of $C$ to $Y$. Since cluster $C$ is $1/3$-well-linked in graph $G$, $|E_C|\geq |E_X|/3$. We charge the edges of $E_C$ for the edges of $E_X$, with the charge to every edge of $E_C$ being at most $3$. Observe that none of the edges in $E_C\cup E_X$ is incident on the source $s$ or the sink $t$.

Let $(\tilde X,\tilde Y)$ be the cut obtained after processing all super-nodes $v_C\in V(G')$.
Let $E_1(X,Y)\sse E(X,Y)$ be the subset of edges incident on the source $s$ or on the sink $t$ in the original cut, and let $E_1(\tilde X,\tilde Y)$ be the subset of edges incident on the source $s$ or on the sink $t$ in the new cut. Since none of the clusters $C$ we have considered contained vertices of $\tset$, or vertices $z_e$ for $e\in E'$, $|E_1(X,Y)|=|E_1(\tilde X,\tilde Y)|$. Let $E_2(X,Y)=E(X,Y)\setminus E_1(X,Y)$, and similarly, let $E_2(\tilde X,\tilde Y)=E(\tilde X,\tilde Y)\setminus E_1(\tilde X,\tilde Y)$. From the above discussion, $|E_2(\tilde X,\tilde Y)|\leq 3|E_2(X,Y)|$. Recall that the capacities of all edges in $E_2(X,Y)$ and $E_2(\tilde X,\tilde Y)$ are $3\eta$, while the capacities of edges in $E_1(X,Y)$ are $1$, and we have assumed that $c(X,Y)=|E_1(X,Y)|+3\eta|E_2(X,Y)|< |\tset'|$.

Cut $(\tilde X,\tilde Y)$ naturally defines an $s$-$t$ cut $(X',Y')$ in network $N_2$. The capacity of this cut in network $N_2$ is $c'(X',Y')=|E_1(X,Y)|+\eta|E_2(\tilde X,\tilde Y)|\leq |E_1(X,Y)|+3\eta|E_2(X,Y)|< |\tset'|$, a contradiction.

%-------------------------------------------------------------------------------------------
\iffalse
%-------------------------------------------------------------------------------------------
\begin{claim}\label{claim: no flow gives contractible set}
Let $G'=(V',E')$ be any graph with unit edge capacities, non-negative weights $w(v)$ on vertices $v\in V'$, and a set $\tset$ of terminals. Let $S\sse V'\setminus\tset$ be any subset of non-terminal vertices, and let $k'$ be any integer. Then either the terminals in $\tset$ can send $k'$ flow units to the edges in $\out(S)$ with no congestion, or we can efficiently find a subset $S'\sse V'\setminus \tset$ of non-terminal vertices, with $|\out_{G'}(S')|<k'$, and $w(S')\geq w(S)$.
\end{claim}

\begin{proof} 
We set up the following flow network. Start with the graph $G'$, and contract the vertices of $\tset$ into a source $s$, and the vertices of $S$ into the sink $t$. We then try to send $k'$ flow units from $s$ to $t$ in this network. If such flow exists, then we return the resulting flow. Otherwise, there is an $s$-$t$ cut $(A,B)$ in our network, with $s\in A$, $t\in B$, and $|E(A,B)|<k'$. Let $S'\sse V'\setminus \tset$ be the set of vertices obtained from $B$, after we replace the sink $t$ with the set $S$ of vertices. Then $|\out(S')|=|E(A,B)|<k'$, and since $S\sse S'$, $w(S')\geq w(S)$.
 \end{proof}

\fi
%--------------------------------------------------------------

\end{document}